\documentclass[11pt]{article}
\usepackage{textcomp}
\usepackage{tabularx,booktabs}
\usepackage[english]{babel}
\usepackage[utf8]{inputenc}
\usepackage{amsmath,amssymb,amsfonts,amsthm,titling}
\usepackage{tikz}
\usepackage{graphicx}
\usepackage{subcaption}
\usepackage{pdfpages}
\usepackage{multicol}
\usepackage{mathtools}
\usepackage{rotating}
\usepackage[colorinlistoftodos]{todonotes}
\usepackage[margin=1in]{geometry}
\usepackage{setspace}
\graphicspath{{img/}}
\usepackage{float}
\usepackage[margin = 1in]{caption}
\usepackage{fancyhdr}
\usepackage{float}
\usepackage{url}
\usepackage[square, numbers]{natbib}
\usepackage{multirow, siunitx, booktabs,array}
\usepackage[font=scriptsize,labelfont=bf]{caption}
\makeatletter 
\def\@eqnnum{{\normalsize \normalcolor (\theequation)}} 
\makeatother

\allowdisplaybreaks
\begin{document}
\begin{titlepage}

\newcommand{\HRule}{\rule{\linewidth}{0.6mm}}
\newcommand{\HRules}{\rule{\linewidth}{0.3mm}} 
\center 

\textsc{\LARGE Partially Specified Space Time Autoregressive Model with Artificial Neural Network}\\[1.5cm] 

 
{Wenqian Wang}\\[0.3cm]
{Beth Andrews}\\[1cm]

{Department of Statistics}\\[0.3cm]
{Northwestern University}\\[1cm]

\end{titlepage}

\section{Introduction}
In Chapter 1, we proposed a PSAR model enhanced by a neural network component which aims at explaining the spatial dependence through a nonlinear approach. However, sometimes we may collect data across time as well as space. For this type of data, we want to construct a model with dependence over time taken into consideration which has a broad application especially in environmental sciences. One interesting application is forecasting the weather. For example, in a fixed location, the everyday temperature will change from time to time but in the meanwhile, it would also be affected by temperatures in the neighboring locations.

A class of such linear models known as space-time autoregressive (STAR) and space-time autoregressive moving average (STARMA) models was introduced by Cliff and Ord (1973) and Martin and Oeppen (1975) in 1970s. 
In general, STAR models contain a hierarchical ordering of ``neighbors" of each site. For instance, on a regular grid, one can categorize neighbors of a site as first-order and second-order neighborhoods and so on. An observation at each site is then modeled as a linear function of the previous time observations at the same site and of the weighted previous observations at the neighboring sites of each order. Let $\{Y_t: t=0,\pm1,\pm2,\ldots\}$ be a multivariate time series of $n$ location components. Weights are incorporated in weight matrices $W^{(k)}$ for order $k$. An STAR model with autoregressive order $p$ and spatial order $(\lambda_1,\ldots, \lambda_p)$ considerded in Borovkova \emph{et.al} (2008) is defined as
\begin{equation*}
Y_t= \sum_{i=1}^{p}\sum_{k=0}^{\lambda_i}\phi_{ik}W^{(k)}Y_{t-i} + \varepsilon_t
\end{equation*}
where $\lambda_i$ is the spatial order of the $i$th autoregressive term, $\phi_{ik}$ is the autoregressive
parameter at time lag $i$ and spatial lag $k$.
Similarly an STAR model with $n$ space locations and $q$ exogenous variables is given by Stoffer (1985) as, for $Y_t\in\mathbb{R}^{n}$,
\begin{equation*}
Y_t= \sum_{i=1}^{p}\sum_{k=0}^{\lambda_i}\phi_{ik}W^{(k)}Y_{t-i}+ \sum_{i=0}^{p^{\prime}}X_{t-i}\beta_i+\varepsilon_t
\end{equation*} 
where values of the exogenous variables $\{X_t: t=0,\pm1,\pm2,\ldots\}$ are $n\times q$ covariate matrices containing $q$ values of exogenous variables for all $n$ locations at time $t$. $X_t =(x_{1,t},\ldots, x_{n,t})^{\prime}$ and $x_{s,t}\in\mathbb{R}^{q}$. $p^{\prime}$ is the autoregressive order for $\{X_t\}$ and $\beta_i$ is a $q\times 1$ model parameter.

STAR models have been widely applied in many areas of science. In genomics, Epperson (1993) analyzed population gene frequencies using STAR models where he assumed genes may vary over space and time. This model is also well known in economics (Giacomini and Granger, 2004) and has been applied to forecasting regional employment (Hernandez and Owyang, 2004) as well as traffic flow (Garrido 2000; Kamarianakis and Prastacos, 2004). For instance, the traffic flow of a road network observed at different fixed locations can be simultaneously modelled as a linear combination of past observations and current observations at neighboring sites. Through weight matrices, an STAR model assumes that near sites exert more influence on each other than distant ones.   

In this chapter, we want to extend an STAR model to a semi-parametric model such that this new model can capture nonlinear dependence between covariates and the spatial observations of interest. 
\section{PSTAR-ANN$(p)$ model}
We define a Partially Specified Space-Time Autoregressive model with Artificial Neural Network (PSTAR-ANN$(p)$) as follows. 
{\small\setlength{\abovedisplayskip}{3pt}
	\setlength{\belowdisplayskip}{\abovedisplayskip}
	\setlength{\abovedisplayshortskip}{0pt}
	\setlength{\belowdisplayshortskip}{3pt}
\begin{gather}\label{pstar-ann}
Y_{t} =  \sum_{i =0}^p \phi_iW_nY_{t-i} + X_{t}\beta+ \boldsymbol{F}(X_{t}\boldsymbol{\gamma}^{\prime})\lambda+\boldsymbol{\varepsilon}_{t},\quad T = 1,\ldots,T
\end{gather}}%
where $Y_t=\{y_{s,t}\}_{s=1}^{n}$ contains observations of dependent variables at $n$ locations and at time $t$. The independent variable matrix $X_{t} =(x_{1,t},\ldots, x_{n,t})^{\prime}$ is the covariate matrix at time $t$, where $x_{s,t}\in \mathbb{R}^{q\times 1}$ is a vector containing exogenous regressors at location $s$ and time $t$, $s=1,\ldots,n$. $\boldsymbol{\varepsilon}_{t} = \{\varepsilon_{s,t}\}_{s=1}^n$ denote a vector of $n$ noise terms which are independent identically distributed across  $s$ and $t$ with density function $f$, mean $0$ and variance $\sigma^2 =1$.

Exogenous parameters $\beta=(\beta_1,\ldots, \beta_q)^{\prime}\in\mathbb{R}^{q}$ and scalars $\phi_{i}, i=0,1,\ldots,p$, the spatial/space-time autoregressive parameters, are assumed to be the same over all regions. $W_n= \{w_{ij}\}\in \mathbb{R}^{n\times n}$ is a known spatial weight matrix which characterizes the connection between neighboring regions. For the ease of illustration, we define some notations. Given a function $f\in C^{1}(R^{1})$ continuous in $\mathbb{R}$, we define a new matrix map $R^{n}\rightarrow R^{n}$ as $\boldsymbol{f}$ s.t. $\boldsymbol{f}(x_1,\ldots,x_n)=(f(x_1),\ldots,f(x_n))^{\prime}$.

Using the notation defined above, the artificial neural network component (Medeiros {\em et al.} \cite{medeiros2006building}) can be written as $ \boldsymbol{F}(X_t\boldsymbol{\gamma}^{\prime})\lambda= $
{\small\setlength{\abovedisplayskip}{3pt}
	\setlength{\belowdisplayskip}{\abovedisplayskip}
	\setlength{\abovedisplayshortskip}{0pt}
	\setlength{\belowdisplayshortskip}{3pt}
	\begin{gather*}
\begin{bmatrix}
	F(x_{1,t}^{\prime}\boldsymbol{\gamma}_1) & F(x_{1,t}^{\prime}\boldsymbol{\gamma}_2)&
	\ldots& F(x_{1,t}^{\prime}\boldsymbol{\gamma}_h)\\
	F(x_{2,t}^{\prime}\boldsymbol{\gamma}_1) & F(x_{2,t}^{\prime}\boldsymbol{\gamma}_2)&
	\ldots& F(x_{2,t}^{\prime}\boldsymbol{\gamma}_h)\\
	\vdots&\vdots&\cdots& \vdots\\
	F(x_{n,t}^{\prime}\boldsymbol{\gamma}_1) & F(x_{n,t}^{\prime}\boldsymbol{\gamma}_2)&
	\ldots& F(x_{n,t}^{\prime}\boldsymbol{\gamma}_h)
\end{bmatrix}
\begin{bmatrix}
\lambda_1\\
\lambda_2\\
\vdots\\
\lambda_h
\end{bmatrix}\in \mathbb{R}^{n}
\end{gather*}}%
$\boldsymbol{F}(X_t\boldsymbol{\gamma}^{\prime})\lambda$ represents two layer NN component where the first layer has $h$-neurons with the sigmoid activation function and the second layer has only one neuron with an identity activation function. In the first layer, the input is $X_{t}$ and weights are $\boldsymbol{\gamma}=(\boldsymbol{\gamma}^{\prime}_1, \ldots, \boldsymbol{\gamma}^{\prime}_h) \in \mathbb{R}^{h\times q}$ where $\boldsymbol{\gamma}_i = (\gamma_{i1},\ldots, \gamma_{iq})^{\prime}$ is the weights in the $i$th neuron. $F(\cdot)$ is the sigmoid activation function in this layer. 
{\small
	\setlength{\abovedisplayskip}{3pt}
	\setlength{\belowdisplayskip}{\abovedisplayskip}
	\setlength{\abovedisplayshortskip}{0pt}
	\setlength{\belowdisplayshortskip}{3pt}
	\begin{equation*}
	F(x_{s,t}^{\prime}\boldsymbol{\gamma}_i)=(1+e^{-x_{s,t}^{\prime}\boldsymbol{\gamma}_i})^{-1}, \quad s=1,2,\ldots,n,\,i=1,2,\ldots,h
	\end{equation*}
} %
In the second layer, the inputs are $F(x_{s,t}^{\prime}\boldsymbol{\gamma}_i),i=1,\ldots, h$ and the weights are $\lambda_1,\ldots,\lambda_h$. So final output is $\sum_{i=1}^{h}\lambda_i F(x_{s,t}^{\prime}\boldsymbol{\gamma}_i)$ for each $x_{s,t}$.

The weight matrix $W_n$ is a measure of distance between the spatial units, and in our application, we begin by using a square symmetric matrix with $(i,j)$ element equals to 1 if regions $i$ and $j$ are neighbors and $0$ otherwise. The diagonal elements of the matrix are set to zero. Then we row standardize this matrix denoted by $W_n$. For more details on construction of the weight matrix, you can refer to the previous chapter or LeSage \cite{lesage1999spatial}.
The following plot provides a preview of the data we are working with. This data is generated from a PSTAR-ANN$(2)$ model in a 10 by 10 lattice. The model equation is shown below:
{\small\setlength{\abovedisplayskip}{3pt}
	\setlength{\belowdisplayskip}{\abovedisplayskip}
	\setlength{\abovedisplayshortskip}{0pt}
	\setlength{\belowdisplayshortskip}{3pt}
	\begin{gather}\label{equ-heatmap}
	Y_t = 0.6W_nY_t -0.274W_nY_{t-1} + X_t\begin{pmatrix}
	0.24\\
	-0.7
	\end{pmatrix} +1.5\boldsymbol{F}(X_t\boldsymbol{\gamma}^{\prime})+\boldsymbol{\varepsilon}_t, \quad X_t=\begin{pmatrix}
	x_{11,t}&\ldots&x_{1n,t}\\
	x_{21,t}&\ldots&x_{2n,t}
	\end{pmatrix}^{\prime}
	\end{gather}}
$\boldsymbol{\gamma}=(0.75,-0.35)$, with $\{x_{1i,t}\}_{i=1}^{n}, \{x_{2i,t}\}_{i=1}^{n}$ are generated i.i.d from $N(0,1.5^2)$, $N(0,3^2)$ and the error $\boldsymbol{\varepsilon}_t$ is from $N(0,1)$. Figure \ref{pstar_example} shows the heatmaps of $Y_t$ simulated at $t=30,29,28$ using (\ref{equ-heatmap}).
\begin{figure}[H]
\begin{center}
  \includegraphics[width=0.95\linewidth]{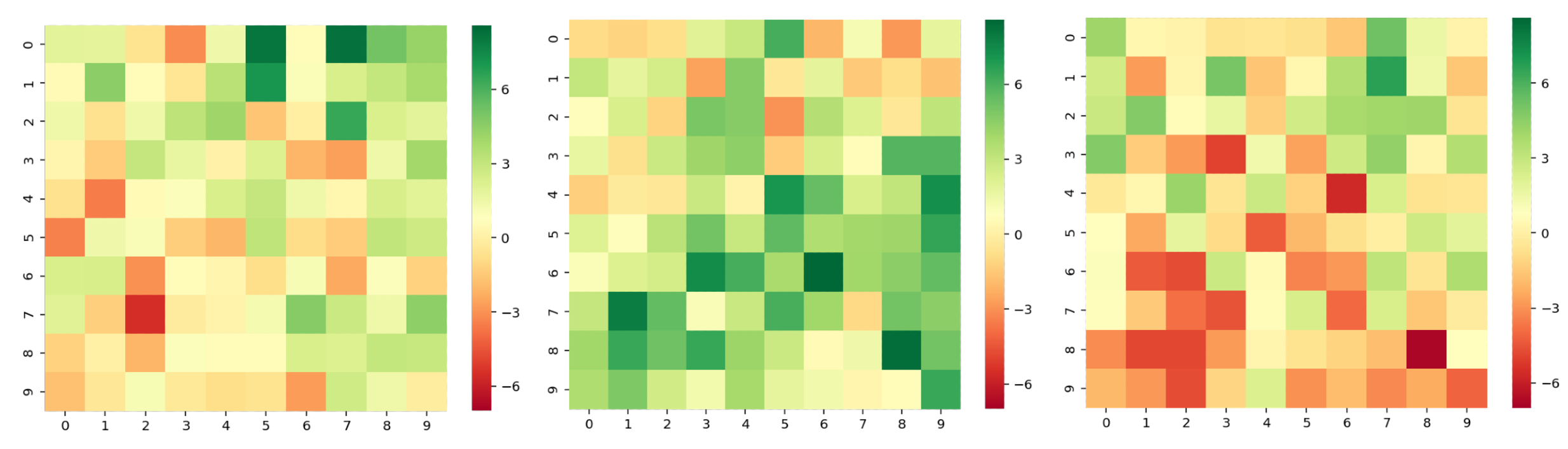}
  \caption{Heat map of $Y_{30}$, $Y_{29}$ and $Y_{28}$ simulated from a PSTAR-ANN$(2)$ model}
  \label{pstar_example}
  \end{center}
\end{figure}
The color scale represents the value in each cell. We can observe colors in cells changing gradually with the spatial and time dependence ($\phi_1 =-0.274$, there is a little flip in cell color comparing the left figure with the middle one).
\theoremstyle{definition}
\newtheorem{pro}{Proposition}
\newtheorem{ass}{Assumption}
\newtheorem{thm}{Theorem}
\newtheorem{lem}{Lemma}
\newtheorem{defn}{Definition}

\section{The Model and the Likelihood Function}
\subsection{The Model}
Let
{\small\setlength{\abovedisplayskip}{3pt}
	\setlength{\belowdisplayskip}{\abovedisplayskip}
	\setlength{\abovedisplayshortskip}{0pt}
	\setlength{\belowdisplayshortskip}{3pt}
\begin{gather*}
A_0 = I_{n}-\phi_0W_n, \quad A_i = \phi_i W_n \quad  i = 1,\ldots,p
\end{gather*}}
Suppose $A_0$ is invertible, then model (\ref{pstar-ann}) can be rewritten as:
{\small\setlength{\abovedisplayskip}{3pt}
	\setlength{\belowdisplayskip}{\abovedisplayskip}
	\setlength{\abovedisplayshortskip}{0pt}
	\setlength{\belowdisplayshortskip}{3pt}
\begin{gather*}
	A_0Y_t = \sum_{i = 1}^{p} A_i Y_{t-i} +X_t\beta +\boldsymbol{F}(X_t\boldsymbol{\gamma}^{\prime})\lambda + \boldsymbol{\varepsilon}_t\\
	Y_t = A_0^{-1}\sum_{i = 1}^{p} A_i Y_{t-i}+A_0^{-1}X_t\beta+ A_0^{-1}\boldsymbol{F}(X_t\boldsymbol{\gamma}^{\prime})\lambda+ A_0^{-1}\boldsymbol{\varepsilon}_t
\end{gather*}}
Let $L$ be the usual backshift operator such that $L^{i}Y_t=Y_{t-i}$, $A(L) = A_0-\sum_{i = 1}^{p}A_iL^{i}$. Assuming that $A^{-1}(L)$ exists, we can rewrite $Y_t$ as
{\small\setlength{\abovedisplayskip}{3pt}
	\setlength{\belowdisplayskip}{\abovedisplayskip}
	\setlength{\abovedisplayshortskip}{0pt}
	\setlength{\belowdisplayshortskip}{3pt}
	\begin{gather}\label{backshift-y}
    Y_t = A(L)^{-1}(X_t\beta +\boldsymbol{F}(X_t\boldsymbol{\gamma}^{\prime})\lambda + \boldsymbol{\varepsilon}_t)
	\end{gather}}
In order to derive asymptotic properties, we also need $Y_t$ to be a causal spatial temporal process. Referring to the definition in Brockwell and Davis \cite{brockwell2002introduction}, the process $Y_t$ is causal if there exists matrices $\{\Psi_j\}$ with absolutely summable components such that $A^{-1}(L)=\sum_{j=0}^{\infty}\Psi_jA_0^{-1}L^j$. Let $A(z) = A_0-A_1z-A_2z^2-\cdots-A_pz^p=A_0(I_n-A_0^{-1}A_1z-A_0^{-1}A_2z^2-\cdots-A_0^{-1}A_pz^p)$ be a matrix-valued polynomial. Causality is equivalent to the condition $det(A(z))\neq0$ for all $z \in \mathbb{C}$ such that $|z|\leq 1$.

The matrices $\Psi_j$ can be found recursively from the equations
{\small\setlength{\abovedisplayskip}{3pt}
	\setlength{\belowdisplayskip}{\abovedisplayskip}
	\setlength{\abovedisplayshortskip}{0pt}
	\setlength{\belowdisplayshortskip}{3pt}
	\begin{gather}\label{backshift-expansion}
	\Psi_j = \Theta_j +\sum_{k=1}^{\infty}A_0^{-1}A_k\Psi_{j-k}
	\end{gather}}
where we define $\Theta_0 = I_n$, $\Theta_j = 0_n$ for $j>0$, $A_j = 0_n$ for $j>p$ and $\Psi_j = 0_n$ for $j <0$. Therefore, this gives us
{\small\setlength{\abovedisplayskip}{3pt}
	\setlength{\belowdisplayskip}{\abovedisplayskip}
	\setlength{\abovedisplayshortskip}{0pt}
	\setlength{\belowdisplayshortskip}{3pt}
	\begin{align*}
	\Psi_0 &=I_n\\
	\Psi_1 &= A_0^{-1}A_1\\
	\Psi_2 &= (A_0^{-1}A_1)^2+ A_0^{-1}A_2\\
	\cdots
	\end{align*}}
Then
{\small\setlength{\abovedisplayskip}{3pt}
	\setlength{\belowdisplayskip}{\abovedisplayskip}
	\setlength{\abovedisplayshortskip}{0pt}
	\setlength{\belowdisplayshortskip}{3pt}
	\begin{align}\label{model-sum}
	Y_t = \sum_{j=0}^{\infty}\Psi_jA_0^{-1}(X_{t-j}\beta+\boldsymbol{F}(X_{t-j}\boldsymbol{\gamma}^{\prime})+\boldsymbol{\varepsilon}_{t-j})
	\end{align}}
With this expansion, we need few assumptions on $\sum_{j=0}^{\infty}\Psi_jA_0^{-1}$ and will be discussed later.
\subsection{Likelihood Function}
Denote $\boldsymbol{\theta}=(\phi_0, \phi_1,\ldots, \phi_p, \beta_1,\ldots,\beta_q, \lambda, \boldsymbol{\gamma}_1^{\prime},\ldots, \boldsymbol{\gamma}_h^{\prime})^{\prime}\in \boldsymbol{\Theta}$. Since $\varepsilon_{s,t}$ has an identical density function $f$, the conditional joint density of $Y_T,Y_{T-1},\ldots,Y_{1} $ conditioned on a finite number of past values $\{Y_{0},\ldots,Y_{1-p}\}$ and $\{X_t\}_{t=1}^{T}$ is 
{\small\setlength{\abovedisplayskip}{3pt}
	\setlength{\belowdisplayskip}{\abovedisplayskip}
	\setlength{\abovedisplayshortskip}{0pt}
	\setlength{\belowdisplayshortskip}{3pt}
	\begin{align*}
	f_{Y_T,Y_{T-1},\ldots,Y_{1}}(\boldsymbol{\theta}|Y_{0},\ldots,Y_{1-p},\{X_t\})=&\prod_{t=1}^{T}f_{Y_t}(\boldsymbol{\theta}|Y_{t-1},\ldots,Y_{1-p},\{X_t\})
	\end{align*}}
Since
{\small\setlength{\abovedisplayskip}{3pt}
	\setlength{\belowdisplayskip}{\abovedisplayskip}
	\setlength{\abovedisplayshortskip}{0pt}
	\setlength{\belowdisplayshortskip}{3pt}
	\begin{align*}
f_{Y_t}(\boldsymbol{\theta}|Y_{t-1},\ldots,Y_{1-p},\{X_t\})=&|A_0|\prod_{s=1}^{n}f(\varepsilon_{s,t}(\boldsymbol{\theta}))
	\end{align*}}
we have 
{\small\setlength{\abovedisplayskip}{3pt}
	\setlength{\belowdisplayskip}{\abovedisplayskip}
	\setlength{\abovedisplayshortskip}{0pt}
	\setlength{\belowdisplayshortskip}{3pt}
	\begin{align*}
f_{Y_T,Y_{T-1},\ldots,Y_{1}}(\boldsymbol{\theta}|Y_{0},\ldots,Y_{1-p},\{X_t\})=&|A_0|^{T}\prod_{t=1}^{T}\prod_{s=1}^{n}f(\varepsilon_{s,t}(\boldsymbol{\theta}))
	\end{align*}}
Hence, the log-likelihood function of $\boldsymbol{\theta}$ is given by \cite[p.~63]{anselin2013spatial},
{\small\setlength{\abovedisplayskip}{3pt}
	\setlength{\belowdisplayskip}{\abovedisplayskip}
	\setlength{\abovedisplayshortskip}{0pt}
	\setlength{\belowdisplayshortskip}{3pt}
	\begin{gather}\label{pstar-loglikelihood}
		\mathcal{L}_{n,T}(\boldsymbol{\theta}) = T\ln |A_0| + \sum_{t=1}^{T}\sum_{s=1}^{n}\ln f(\varepsilon_{s,t}(\boldsymbol{\theta}))
	\end{gather}}
where $\boldsymbol{\varepsilon}_t(\boldsymbol{\theta})=\{\varepsilon_{s,t}(\boldsymbol{\theta})\}_{s=1}^n= A(L)Y_t- X_{t}\beta-\boldsymbol{F}(X_{t}\boldsymbol{\gamma})\lambda$ for $t=1,\ldots,T$.

For the analysis of identification and estimation of the PSTAR-ANN$(p)$ model, we adopt the following assumptions:
\begin{ass}\label{compact-space}
	The $p + (q+1)(h+1)$ parameter vector $\boldsymbol{\theta}=(\phi_0, \phi_1,\ldots, \phi_p,\beta^{\prime},\lambda^{\prime},\boldsymbol{\gamma}_1^{\prime},\ldots,\boldsymbol{\gamma}_h^{\prime})^{\prime} \in \boldsymbol{\Theta}$, where $\boldsymbol{\Theta}$ is a subset of the $p+(q+1)(h+1)$ dimensional Euclidean space, $\mathbb{R}^{p+(q+1)(h+1)}$. $\boldsymbol{\Theta}$ is a closed and bounded compact set and contains the true parameter value $\boldsymbol{\theta}_0$ as an interior point.
\end{ass}
\begin{ass}\label{phi0-range}
	The spatial correlation coefficient $\phi_0$ satisfies $|\phi_0|<1$ and
	$\phi_0 \in (-1/\tau,1/\tau)$, where $\tau= max\{|\tau_1|,|\tau_2|,\ldots,|\tau_n|\}$, $\tau_1,\ldots, \tau_n$ are eigenvalues of spatial weight matrix $W_n$. To avoid the non-stationarity issue when $\phi_0$ approaches to 1, we assume $\sup_{\phi_0\in\boldsymbol{\Theta}}|\phi_0|<1$.
\end{ass}
\begin{ass}\label{ub-weight-matrix}
	We assume $W_n$ is defined by queen contiguity and is uniformly bounded in row and column sums in absolute value as $n\rightarrow\infty$ so $A_0^{-1}$ is also uniformly bounded in both column and row sums as $n\rightarrow\infty$. 
\end{ass}
\begin{ass}\label{ub-backshift-matrix}
	We assume a causal spatial process $Y_t$ which means that every $z$ which solves
	{\small\setlength{\abovedisplayskip}{3pt}
		\setlength{\belowdisplayskip}{\abovedisplayskip}
		\setlength{\abovedisplayshortskip}{0pt}
		\setlength{\belowdisplayshortskip}{3pt}
		\begin{gather*}
		\det \left[ z^pA_0-\sum_{i=1}^{p}\phi_iW_nz^{p-i} \right]=0
		\end{gather*}}%
	lie inside a unit circle. So the operator $A(L)$ is causal \cite{pfeifer1980three}.
\end{ass}
\begin{ass}\label{x-boundness}
	$X_t$ is stationary, ergodic satisfying $\mathbb{E}\,|x_{s,t}|^2 < \infty$ and $X_t$ is full column rank for $t=1,2\ldots,T$.
\end{ass}
\begin{ass}\label{error-dist}
	The error terms $\varepsilon_{s,t}$, $s=1,2,\ldots,n$, $t=1,2\ldots,T$ are independent and identically distributed with density function $f(\cdot)$, zero mean and unit variance $\sigma^2=1$. The moment $\mathbb{E}(|\varepsilon_{s,t}|^{2+r})$ exists for some $r>0$ and $\mathbb{E}|\ln f(\varepsilon_{s,t})|<\infty$.
\end{ass}
Assumption \ref{phi0-range} defines the parameter space for $\phi_0$ such that $A_0$ is strictly diagonally dominant. By the Levy-Desplanques theorem \cite{taussky1949recurring}, it follows that $A_0^{-1}$ exists for any values $\phi_0$ in $(-1/\tau, 1/\tau)$. In real applications, since $W_n$ is row standardized, one just searches $\hat{\phi}_0$ over a parameter space on $(-1,1)$ to find the optimizer \cite[p.~749-754]{gershgorin1931uber}.

It is natural to consider the neighborhood by connections and in many practical studies, since entries scaled to sum up to 1, each row of $W_n$ sums up to 1, which guarantees that all nonzero weights are in $(0,1]$. For simplicity, we define the weight matrix $W_n$ using the queen criterion and do row standardization. Assumption \ref{ub-weight-matrix} is originated by Kelejian and Prucha \cite{kelejian1998generalized,kelejian1999generalized} and is also used in Lee \cite{lee2004asymptotic}. 
With $W_n$ to be uniformly bounded, we can prove that $(I_n - \phi_0 W_n)^{-1}$ is also uniformly bounded in row and column sums for $\phi_0 \in (-1/\tau,1/\tau)$ and $\sup_{\phi_0\in\boldsymbol{\Theta}}|\phi_0|<1$, by Lemma A.4 in Lee\cite{lee2004asymptotic}. This result is a necessary condition for Assumption \ref{ub-backshift-matrix}.

From Assumption \ref{phi0-range} and \ref{ub-weight-matrix}, we can decompose $W_n$ by its eigenvalue and eigenvector pairs $\tau_i, v_i$: $W_n = P\Lambda P^{-1}$, where $\Lambda$ is a diagonal matrix with eigenvalues $\tau_i$ on its diagonals and $P=[v_1,v_2,\ldots, v_n]$ (we assume $v_i$'s are normalized eigenvectors). So
{\small\setlength{\abovedisplayskip}{3pt}
	\setlength{\belowdisplayskip}{\abovedisplayskip}
	\setlength{\abovedisplayshortskip}{0pt}
	\setlength{\belowdisplayshortskip}{3pt}
	\begin{align}\label{matrix-decomposition}
	W=P\begin{pmatrix}
	\tau_1& 0 &\cdots&0\\
	0& \tau_2&\cdots&0\\
	0&0&\ddots&0\\
	0&0 &\cdots&\tau_n
	\end{pmatrix}P^{-1},
		A_0^{-1}=P\begin{pmatrix}
		\frac{1}{1-\phi_0\tau_1}& 0 &\cdots&0\\
		0& \frac{1}{1-\phi_0\tau_2}&\cdots&0\\
		0&0&\ddots&0\\
		0&0 &\cdots&\frac{1}{1-\phi_0\tau_n}
		\end{pmatrix}P^{-1}
	\end{align}}%
It is trivial that $A_0^{-1}W_n = W_nA_0^{-1}$.

Assumption \ref{ub-backshift-matrix} guarantees that $A(L)$ is a causal operator and there exists a casual solution $\{Y_t\}$ to the system of the model equation (\ref{pstar-ann}). Then $\sum_{j=0}^{\infty}\Psi_j A_0^{-1}$ is absolutely summable. This requirement serves to determine a region of possible $\phi_i$ values that will result in a stationary process $\{Y_t\}$.

Assumption \ref{x-boundness} is a trivial one when exogenous variables are included in a space time model. Similar to previous chapter, the stationarity of $\{x_{s,t}\}$ is necessary in the ergodic theorem in later proofs.

Assumption \ref{error-dist} imposes restrictions for the random error. In this paper we mainly consider the heavy tailed density functions such scaled $t$ distributions and Laplace distributions. When the degrees of freedom goes to infinity, the scaled $t$ distribution would approximate a standard normal distribution. So we would like to concentrate more on the scaled $t$ distribution with lower degrees of freedom.

\section{Model Identification}
In the previous section, we have some restrictions on the weight matrices $W_n$ and $A_{i}$'s to guarantee the identification of a classical spatial time autoregressive model. We now investigate the conditions under which PSTAR(p)-ANN model is identified. By Rothenberg \cite{rothenberg1971identification}, a parameter $\theta_0\in\boldsymbol{\Theta}$ is {\it globally identified} if there is no other $\theta$ in $\boldsymbol{\Theta}$ that observationally equivalent to $\theta_0$ such that $f(y,\theta)=f(y,\theta_0)$; or the parameter $\theta_0$ is {\it locally identified} if there is no such $\theta$ in an open neighborhood of $\theta_0$ in $\boldsymbol{\Theta}$.  
The model (\ref{pstar-ann}), in principle, is neither globally nor locally identified due to the neural network component. The lack of identification of neural network models has been discussed in many papers (Hwang and Ding \cite{hwang1997prediction}; Medeiros \emph{et al.} \cite{medeiros2006building}). Here we extend the discussion to our proposed PSTAR(p)-ANN model. Three characteristics imply non-identification of our model:
(a) the interchangeable property: the value of the likelihood function may remain unchanged if we permute the hidden units. For a model with $h$ neurons, this will result in $h!$ different models that are indistinguishable from each other and have equal local maximums of the log-likelihood function;
(b) the ``symmetry" property: for a logistic function, $F(x)=1-F(-x)$ allows two equivalent parametrization for each hidden unit;
(c) the reducible property: the presence of irrelevant neurons in model (\ref{pstar-ann}) happens when $\lambda_i=0$ for at least one $i$ and parameters $\boldsymbol{\gamma}_i$ remain unidentified. Conversely, if $\boldsymbol{\gamma}_{i} =\mathbf{0}$, $F(X_t\boldsymbol{\gamma}_i)$ is a constant and $\lambda_i$ can take any value without affecting the value of likelihood functions. 

The problem of interchangeability (as mentioned in (a)) can be solved by imposing the following restriction, as in Medeiros \emph{et al.} \cite{medeiros2006building}:\\
{\bf Restriction 1.} {\it parameters $\lambda_1, \ldots, \lambda_h $  are restricted such that: $\lambda_1\geq\cdots\geq \lambda_h$.
}\\
And to tackle (b) and (c), we can apply another restriction:\\
{\bf Restriction 2.} {\it The parameters $\lambda_i$ and $\gamma_{i1}$ should satisfy:\\
	(1) $\lambda_i \neq 0$, $\forall i \in \{1,2,\ldots, h\}$; and\\
	(2) $\gamma_{i1} > 0$, $\forall i \in \{1,2,\ldots, h\}$.      
}\\
To guarantee the non-singularity of model matrices and the uniqueness of parameters, we impose the following basic assumption:   
\begin{ass}\label{x-fullrank}
	The true parameter vector $\boldsymbol{\theta}_0$ satisfies Restrictions 1-2.
\end{ass}
Referring to the section 4.3 by Medeiros \emph{et al.} \cite{medeiros2006building}, we can conclude the identifiability of the PSAR-ANN model.
\begin{lem}\label{identify}
	Under the Assumptions \ref{compact-space}-\ref{x-fullrank}, this PSTAR-ANN$(p)$ model (\ref{pstar-ann}) is globally identified.
\end{lem}

\section{Asymptotic Results}
Let the true parameter vector as $\boldsymbol{\theta}_0$ and the solution which maximizes the log-likelihood function (\ref{pstar-loglikelihood}) as $
\hat{\boldsymbol{\theta}}_{n,T}$. Hence, $\hat{\boldsymbol{\theta}}_{n,T}$ should satisfy
{\small\setlength{\abovedisplayskip}{3pt}
	\setlength{\belowdisplayskip}{\abovedisplayskip}
	\setlength{\abovedisplayshortskip}{0pt}
	\setlength{\belowdisplayshortskip}{3pt}
\begin{align*}
\hat{\boldsymbol{\theta}}_{n,T} = \arg\max_{\boldsymbol{\theta}\in \boldsymbol{\Theta}} \mathcal{L}_{n,T}(\boldsymbol{\theta})
\end{align*}}%
Suppose as $n$ is large enough, $T$ goes to infinity, $\hat{\boldsymbol{\theta}}_{n,T}$ is equivalent to maximizing the average of the likelihood function $\mathcal{L}_{n,T}(\boldsymbol{\theta})$ shown as follows:
\begin{align*}
\frac{1}{nT} \mathcal{L}_{n,T}(\boldsymbol{\theta}) &= \frac{1}{n}\ln |A_0| + \frac{1}{nT} \sum_{s=1}^{n}\sum_{t=1}^{T}\ln f(\varepsilon_{s,t}(\boldsymbol{\theta}))\\
\hat{\boldsymbol{\theta}}_{n,T} &= \arg\max_{\boldsymbol{\theta}\in \boldsymbol{\Theta}}  \left(\frac{1}{n}\ln |A_0| + \frac{1}{nT} \sum_{s=1}^{n}\sum_{t=1}^{T}\ln f(\varepsilon_{s,t}(\boldsymbol{\theta}))\right)\\
\varepsilon_{s,t}(\boldsymbol{\theta}) &= y_{s,t}-\sum_{i=0}^{p}\sum_{k=1}^{n}\phi_iw_{sk}y_{k,t-i}-x_{s,t}^{\prime}\beta-\sum_{i=1}^{h}\lambda_i F(x_{s,t}^{\prime}\boldsymbol{\gamma}_i)
\end{align*}
At specific time $t$, suppose we have a $n_1 \times n_2$ lattice where we consider asymptotic properties of $\hat{\boldsymbol{\theta}}_{n,T}$ when $n=n_1n_2\rightarrow\infty$. Write the location $s$ as the coordinate $(s_x, s_y)$ in the $[1,n_1]\times[1,n_2]$ lattice space. The distance between two locations $s,j$ is defined as $d(s,j) = \max(|s_x-j_x|, |s_y-j_y|)$. So if observations at $s,j$ locations are neighbors (by queen criterion), their coordinates should satisfy $(s_x-j_x)^2+(s_y-j_y)^2\leq2$ or $d(s,j) =1$.

In a spatial context, we should notice that the functional form of $y_{s,t}$ is not identical for all the locations due to values of the weights $\{w_{si}\}_{i=1}^n$. For example, in a lattice, units at edges, vertexes or in the interior have different density functions due to different neighborhood structures (Figure \ref{neigh-str}). Denote $\mathcal{N}_s$ as a neighborhood set for location $s$. For an interior point (Figure \ref{neigh-str}(c)), its neighborhood set $\mathcal{N}_s$ contains eight neighbors where $w_{sj}=1/8$ if $d(s,j)=1$ otherwise $w_{sj}=0$, for $j=1,2,\ldots, n$. Similarly, an edge point (Figure \ref{neigh-str}(b)) has five neighboring units with $w_{sj}=1/5$ for $j\in \mathcal{N}_s$ and the weight of a vertex neighborhood is $1/3$ because a vertex unit has only three neighbors. This is known as an edge effect in spatial problems.
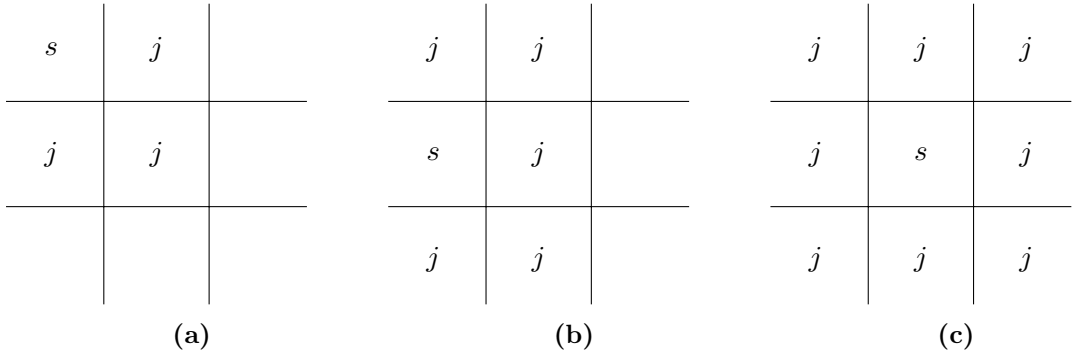
\begin{figure}[ht]
	\centering
	\begin{subfigure}[b]{0.3\textwidth}
		\begin{tikzpicture}
		\draw[step=1.4cm,color=black] (0.1,0.1) grid (4.1,4.1);
		\node at (0.7,3.5) {$s$};
		\node at (2.1,3.5) {$j$};
		\node at (0.7,2.1) {$j$};
		\node at (2.1,2.1) {$j$};
		\end{tikzpicture}
		\caption{}
	\end{subfigure}
	\begin{subfigure}[b]{0.3\textwidth}
		\begin{tikzpicture}
		\draw[step=1.4cm,color=black] (0.1,0.1) grid (4.1,4.1);
		\node at (0.7,0.7) {$j$};
		\node at (2.1,0.7) {$j$};
		\node at (2.1,2.1) {$j$};
		\node at (0.7,2.1) {$s$};
		\node at (2.1,3.5) {$j$};
		\node at (0.7,3.5) {$j$};
		\end{tikzpicture}
		\caption{}
	\end{subfigure}
	\begin{subfigure}[b]{0.3\textwidth}
		\begin{tikzpicture}
		\draw[step=1.4cm,color=black] (0.1,0.1) grid (4.1,4.1);
		\node at (0.7,0.7) {$j$};
		\node at (2.1,0.7) {$j$};
		\node at (3.5,0.7) {$j$};
		\node at (0.7,2.1) {$j$};
		\node at (2.1,2.1) {$s$};
		\node at (3.5,2.1) {$j$};
		\node at (0.7,3.5) {$j$};
		\node at (2.1,3.5) {$j$};
		\node at (3.5,3.5) {$j$};
		\end{tikzpicture}
		\caption{}
	\end{subfigure}
	\caption{Vertex (a), Edge (b) and Interior Points (c) Neighborhood Structures: $s$ is the target location and $j$ represents the neighborhood of $s$}
	\label{neigh-str}
\end{figure}

\noindent
To deal with this, referring to Yao and Brockwell \cite{yao2006gaussian}, we construct an edge effect correction scheme based on the way that the sample size tends to infinity. 
In a space $[1,n_1]\times [1,n_2]$, we consider its interior area as $\mathcal{S}=\{(s_x,s_y): b_1\leq s_x \leq n_1-b_1, b_2\leq s_y \leq n_2-b_2\}$, where $b_1,b_2,n_1,n_2\rightarrow\infty$ satisfying that $b_1/n_1,b_2/n_2\rightarrow0$ and other locations belong to the boundary areas $\mathcal{M}$.
Therefore the set $\mathcal{S}$ contains $n^{\ast}=(n_1-2b_1)(n_2-2b_2)$ interior locations while the set $\mathcal{M}$ contains $n-n^{\ast}$ boundary locations. Then $n^{\ast}/n\rightarrow1$ and $\mathcal{L}_{n,T}(\boldsymbol{\theta})$ can be split into a sum of two parts (interior $\mathcal{S}$ and boundary $\mathcal{M}$ parts):
{\small\setlength{\abovedisplayskip}{3pt}
	\setlength{\belowdisplayskip}{\abovedisplayskip}
	\setlength{\abovedisplayshortskip}{0pt}
	\setlength{\belowdisplayshortskip}{3pt}
	\begin{align*}
\mathcal{L}_{n,T}(\boldsymbol{\theta})=&\sum_{t=1}^{T}\left(\sum_{s\in \mathcal{M}}l(\boldsymbol{\theta}|z_{s,t})+\sum_{s\in \mathcal{S}}l(\boldsymbol{\theta}|z_{s,t})\right)\\
l(\boldsymbol{\theta}|z_{s,t})=&\frac{1}{n}\ln |A_0|+\ln f(y_{s,t}-\sum_{i=0}^{p}\sum_{k=1}^{n}\phi_iw_{sk}y_{k,t-i}-x_{s,t}^{\prime}\beta-\sum_{i=1}^{h}\lambda_i F(x_{s,t}^{\prime}\boldsymbol{\gamma}_i))
	\end{align*}} %
where $Z_t=(W_nY_t,W_nY_{t-1}, \ldots, W_nY_{t-p}, X_t)$ and $z_{s,t}$ is the $s$ row of $Z_t$. 

Therefore, given that $\lim_{n_1,n_2\rightarrow\infty}\frac{|\mathcal{M}|}{n}=0$, $n^{-1}\sum_{s\in \mathcal{M}}l(\boldsymbol{\theta}|z_{s,t})$ vanishes a.s. as $n$ tends to infinity for any $\boldsymbol{\theta}\in\boldsymbol{\Theta}$. Therefore,
{\small\setlength{\abovedisplayskip}{3pt}
	\setlength{\belowdisplayskip}{\abovedisplayskip}
	\setlength{\abovedisplayshortskip}{0pt}
	\setlength{\belowdisplayshortskip}{3pt}
	\begin{align*}
	\lim_{n,T\rightarrow\infty}(nT)^{-1} \mathcal{L}_{n,T}(\boldsymbol{\theta})&=\lim_{T\rightarrow\infty}\frac{1}{T}\sum_{t=1}^{T}\lim_{n_1,n_2\rightarrow\infty} \frac{1}{n_1n_2}\left(\sum_{s\in \mathcal{M}}l(\boldsymbol{\theta}|z_{s,t})+\sum_{s\in \mathcal{S}}l(\boldsymbol{\theta}|z_{s,t})\right)\\
&= \lim_{T\rightarrow\infty}\frac{1}{T}\sum_{i=1}^{T}\lim_{n_1,n_2\rightarrow\infty}\frac{1}{n_1n_2}\sum_{s\in \mathcal{S}}l(\boldsymbol{\theta}|z_{s,t})\quad a.s.
	\end{align*}} %
In this equation, every location $s\in \mathcal{S}$ has eight neighboring units under the queen criterion with nonzero weights $w_{sj}=1/8$. Hence for an interior unit $s\in\mathcal{S}$, $\sum_{i=1}^{n}w_{si}y_i= \sum_{j=1}^n\frac{1}{8}y_jI_{\{d(s,j)=1\}}$. And the log likelihood function $\mathcal{L}_{n,T}(\boldsymbol{\theta})$ is approximately
{\small\setlength{\abovedisplayskip}{3pt}
	\setlength{\belowdisplayskip}{\abovedisplayskip}
	\setlength{\abovedisplayshortskip}{0pt}
	\setlength{\belowdisplayshortskip}{3pt}
	\begin{equation}\label{mod-average}
	(nT)^{-1}\mathcal{L}_{n,T}(\boldsymbol{\theta})\approx \frac{1}{nT}\sum_{i=1}^{T}\sum_{s\in \mathcal{S}}l(\boldsymbol{\theta}|z_{s,t}) \quad \text{for } n_1,n_2,T\rightarrow\infty
	\end{equation}} %
So the maximum likelihood estimator $\hat{\boldsymbol{\theta}}_{n,T}$ approximately maximizes
{\small\setlength{\abovedisplayskip}{3pt}
	\setlength{\belowdisplayskip}{\abovedisplayskip}
	\setlength{\abovedisplayshortskip}{0pt}
	\setlength{\belowdisplayshortskip}{3pt}
	\begin{gather*}
	\hat{\boldsymbol{\theta}}_{n,T} \approx \arg\max_{\boldsymbol{\theta}\in\boldsymbol{\Theta}}\lim_{\substack{T\rightarrow\infty\\ n_1,n_2\rightarrow\infty}}\frac{1}{nT}\sum_{i=1}^{T}\sum_{s\in \mathcal{S}}l(\boldsymbol{\theta}|z_{s,t})
	\end{gather*}} %
\subsection{Consistency Results}
To establish the consistency of $\hat{\boldsymbol{\theta}}_{n,T}$, the heuristic insight is that because $\hat{\boldsymbol{\theta}}_{n,T}$ maximizes $\frac{1}{nT}\mathcal{L}_{n,T}(\boldsymbol{\theta})$, it approximately maximizes $\frac{1}{nT}\sum_{i=1}^{T}\sum_{s\in \mathcal{S}}l(\boldsymbol{\theta}|z_{s,t})$. By equation (\ref{mod-average}), $\frac{1}{nT}\mathcal{L}_{n,T}(\boldsymbol{\theta})$ can generally be shown tending to a real function $\mathcal{L}:\boldsymbol{\Theta}\rightarrow\mathbb{R}$ with maximizer $\boldsymbol{\theta}_0$ as $n, T\rightarrow\infty$ under mild conditions on the data generating process, then $\hat{\boldsymbol{\theta}}_{n,T}$ should tend to $\boldsymbol{\theta}_0$ almost surely. Before the formal proof of the consistency, we need the following assumptions on the density function $f(\cdot)$ satisfied (similar assumptions are made in White \cite{white1996estimation}, Andrews, Davis and Breidt \cite{andrews2006maximum}, Lii and Rosenblatt \cite{lii1992approximate}).
\begin{ass}\label{continuous}
	For all $s\in \mathbb{R}$, $f(s) >0 $ and $f(s)$ is twice continuously differentiable with respect to $s$.
\end{ass}
\begin{ass}\label{error-integral}
	The density should satisfy the following equations:\\
	$\bullet$ $\int sf^{\prime}(s)\,ds = sf(s)|^{\infty}_{-\infty}-\int f(s)\,ds = -1$\\
	$\bullet$ $\int f^{\prime\prime}(s)\,ds = f^{\prime}(s)|^{\infty}_{-\infty} = 0$\\
	$\bullet$ $\int s^2f^{\prime\prime}(s)\,ds = s^2f^{\prime}(s)|^{\infty}_{-\infty} - 2\int sf^{\prime}(s)\,ds = 2$
\end{ass}
\begin{ass}\label{error-dominance}
	The density should follow the following dominance conditions:\\
	$\left|\frac{f^{\prime}(s)}{f(s)}\right|$, $\left|\frac{f^{\prime}(s)}{f^(s)}\right|^2$,  $\left|\frac{f^{\prime}(s)}{f(s)}\right|^4$, $\frac{f^{\prime\prime}(s)}{f(s)}$,  and $\frac{f^{\prime\prime}(s)f^{\prime2}(s)}{f^3(s)}$ are dominated by $a_1+a_2\left|s\right|^{c_1}$, where $a_1$, $a_2$, $c_1$ are non-negative constants and $\int \left|s\right|^{c_1+2}f(s)\,ds < \infty$.
\end{ass}
\begin{ass}\label{x-moment}
	If $c_1>2$ in previous assumption, we further assume $\mathbb{E}\,|x_{s,t}|^{c_1}<\infty$.
\end{ass}
Discussed in Breidt, Davis, Lii and Rosenblatt \cite{breid1991maximum} and Andrews, Davis and Breidt \cite[p.~1642-1645]{andrews2006maximum}, these assumptions on the density $f(\cdot)$ are satisfied by the t-distribution case when $\nu >2$ and by a mixture of Gaussian distributions. The assumption $\mathbb{E}|\ln f(s)|<\infty$ (see Assumption \ref{error-dist}) is also checked satisfied by the normal and t distributions ($\nu>2$). The Laplace distribution does not strictly satisfy the Assumptions \ref{continuous}-\ref{error-dominance}, since it is not differentiable at 0 but it satisfies these boundedness conditions almost everywhere so we believe the consistency and asymptotic normality results remain valid for parameter estimates. This will be shown in the simulation section. Assumption \ref{x-moment} is a necessary to boundedness conditions in later proof.
\begin{lem}\label{unique-max}
	Given Assumptions \ref{compact-space}-\ref{error-dominance}, 
		{\small\setlength{\abovedisplayskip}{3pt}
		\setlength{\belowdisplayskip}{\abovedisplayskip}
		\setlength{\abovedisplayshortskip}{0pt}
		\setlength{\belowdisplayshortskip}{3pt}
		\begin{align*}
		\boldsymbol{\theta}_0=\max_{\boldsymbol{\theta}\in\boldsymbol{\Theta}}\mathbb{E}\,\mathcal{L}_{n,T}(\boldsymbol{\theta})\equiv\max_{\boldsymbol{\theta}\in\boldsymbol{\Theta}}\mathbb{E}\,\frac{1}{nT}\mathcal{L}_{n,T}(\boldsymbol{\theta})
		\end{align*}}%
\end{lem}
\begin{proof}
	$L_{n,T}$ is the joint density function of $Y_t, X_t$ for $t=1,\ldots,T$.
	{\small\setlength{\abovedisplayskip}{3pt}
		\setlength{\belowdisplayskip}{\abovedisplayskip}
		\setlength{\abovedisplayshortskip}{0pt}
		\setlength{\belowdisplayshortskip}{3pt}
		\begin{gather*}
		\mathbb{E}\,\mathcal{L}_{n,T}(\boldsymbol{\theta})-\mathbb{E}\,\mathcal{L}_{n,T}(\boldsymbol{\theta}_0)=\mathbb{E}\,\ln\frac{L_{n,T}(\boldsymbol{\theta})}{L_{n,T}(\boldsymbol{\theta}_0)}
		\end{gather*}}%
Denote $Z=(Y_T,X_T,\ldots, Y_1,X_1)$. By Jensen's inequality,
	{\small\setlength{\abovedisplayskip}{3pt}
		\setlength{\belowdisplayskip}{\abovedisplayskip}
		\setlength{\abovedisplayshortskip}{0pt}
		\setlength{\belowdisplayshortskip}{3pt}
		\begin{gather*}
		\mathbb{E}\,\ln\frac{L_{n,T}(\boldsymbol{\theta})}{L_{n,T}(\boldsymbol{\theta}_0)}\leq \ln \mathbb{E}\,\frac{L_{n,T}(\boldsymbol{\theta})}{L_{n,T}(\boldsymbol{\theta}_0)}=\ln \int_{-\infty}^{\infty}\frac{L_{n,T}(\boldsymbol{\theta})}{L_{n,T}(\boldsymbol{\theta}_0)}L_{n,T}(\boldsymbol{\theta}_0)\,dZ=0
		\end{gather*}}%
	So $\mathbb{E}\,\mathcal{L}_{n,T}(\boldsymbol{\theta})<\mathbb{E}\,\mathcal{L}_{n,T}(\boldsymbol{\theta}_0)$. By Lemma \ref{identify}, the PSTAR(p)-ANN model is globally identified and therefore $\mathbb{E}\,\mathcal{L}_{n,T}(\boldsymbol{\theta})$ is uniquely maximized at $\boldsymbol{\theta}_0$ for all $n,T$. Since the parameter vector $\boldsymbol{\theta}$ does not depend on $n$ and $T$, it is equivalent to say that $\boldsymbol{\theta}_0=\max_{\boldsymbol{\theta}\in \boldsymbol{\Theta}}\mathbb{E}\,\frac{1}{nT}\mathcal{L}_{n,T}(\boldsymbol{\theta})$.
\end{proof}
We define a Hadamard product denoted by $\circ$, s.t. for vectors $a,b_1,\ldots,b_n\in\mathbb{R}^n$, a matrix $B=(b_1, \ldots,b_n)\in \mathbb{R}^{n \times n}$,
{\small\setlength{\abovedisplayskip}{3pt}
	\setlength{\belowdisplayskip}{\abovedisplayskip}
	\setlength{\abovedisplayshortskip}{0pt}
	\setlength{\belowdisplayshortskip}{3pt}
	\begin{gather*}
	a\circ B =\begin{bmatrix}
	a_1b_{11} & a_1b_{21}&\cdots&a_1b_{n1}\\
	a_2b_{12} & a_2b_{22}&\cdots&a_2b_{n2}\\
	\vdots&\vdots&\ddots&\vdots\\
	a_nb_{1n} & a_nb_{2n}&\cdots&a_nb_{nn}
	\end{bmatrix},
	a\circ b_1 =\begin{bmatrix}
	a_1b_{11}\\
	a_2b_{12}\\
	\vdots\\
	a_nb_{1n}
	\end{bmatrix}
	\end{gather*}}%
And let
{\small\setlength{\abovedisplayskip}{3pt}
	\setlength{\belowdisplayskip}{\abovedisplayskip}
	\setlength{\abovedisplayshortskip}{0pt}
	\setlength{\belowdisplayshortskip}{3pt}
	\begin{align*}
	k_0 &= \int \left|\frac{f^{\prime}(s)}{f(s)}\right|f(s)\,ds\\
	k_1 &= \int \left|\frac{f^{\prime^2}(s)}{f^2(s)}-\frac{f^{\prime\prime}(s)}{f(s)}\right|f(s)\,ds\\
	k_2 & =\int \left|\frac{sf^{\prime^2}(s)}{f(s)}-\frac{sf^{\prime\prime}(s)}{f(s)}\right|f(s)\,ds\\
	k_3 &= \int \left|\frac{s^2f^{\prime^2}(s)}{f(s)}-\frac{s^2f^{\prime\prime}(s)}{f(s)}\right|f(s)\,ds
	\end{align*}} %
To facilitate the proof later on, we provide a lemma as follows.
\begin{lem}\label{uni-cvg}
	Given Assumptions \ref{compact-space}-\ref{x-moment}, 
	{\small\setlength{\abovedisplayskip}{3pt}
		\setlength{\belowdisplayskip}{\abovedisplayskip}
		\setlength{\abovedisplayshortskip}{0pt}
		\setlength{\belowdisplayshortskip}{3pt}
		\begin{equation}\label{uni-cvg-eqn}
		\sup\limits_{\boldsymbol{\theta}\in\boldsymbol{\Theta}}\left|\frac{1}{nT}\sum_{s=1}^{n}\sum_{t=1}^{T}\ln f(\varepsilon_{s,t}(\boldsymbol{\theta}))-\mathbb{E}\,\frac{1}{nT}\sum_{s=1}^{n}\sum_{t=1}^{T}\ln f(\varepsilon_{s,t}(\boldsymbol{\theta}))\right|\xrightarrow{p} 0\text{ as }n,T\rightarrow\infty
		\end{equation}} %
\end{lem}
\begin{proof}
	As illustrated in equation (\ref{mod-average}), in a lattice with size $n_1\times n_2$,
	{\small\setlength{\abovedisplayskip}{3pt}
		\setlength{\belowdisplayskip}{\abovedisplayskip}
		\setlength{\abovedisplayshortskip}{0pt}
		\setlength{\belowdisplayshortskip}{3pt}
		\begin{gather*}
		\sup\limits_{\boldsymbol{\theta}\in\boldsymbol{\Theta}}\left|\frac{1}{nT}\sum_{t=1}^{T}\sum_{s=1}^{n}\ln f(\varepsilon_{s,t}(\boldsymbol{\theta}))-\frac{1}{nT}\sum_{t=1}^{T}\sum_{s\in\mathcal{S}}\ln f(\varepsilon_{s,t}(\boldsymbol{\theta}))\right|\xrightarrow{a.s.}0\text{ as }n_1,n_2, T\rightarrow\infty
		\end{gather*}} %
Therefore, to prove (\ref{uni-cvg-eqn}) is equivalent to show that
{\small\setlength{\abovedisplayskip}{3pt}
	\setlength{\belowdisplayskip}{\abovedisplayskip}
	\setlength{\abovedisplayshortskip}{0pt}
	\setlength{\belowdisplayshortskip}{3pt}
	\begin{gather}\label{uni-cvg-eqn-2}
	\sup\limits_{\boldsymbol{\theta}\in\boldsymbol{\Theta}}\left|\frac{1}{nT}\sum_{t=1}^{T}\left(\sum_{s\in\mathcal{S}}\ln f(\varepsilon_{s,t}(\boldsymbol{\theta}))-\mathbb{E}\, \frac{1}{n}\sum_{s\in\mathcal{S}}\ln f(\varepsilon_{s,t}(\boldsymbol{\theta}))\right)\right|\xrightarrow{p} 0\text{ as }n_1,n_2,T\rightarrow\infty
	\end{gather}} %
where $\mathcal{S}$ denotes the interior units mentioned before. Since the interior units have the same neighboring structure, the space process for them is stationary when $n_1,n_2$ go to infinity.
We first show $\left|\frac{1}{nT}\sum_{t=1}^{T}\left(\sum_{s\in\mathcal{S}}\ln f(\varepsilon_s(\boldsymbol{\theta}))-\mathbb{E}\, \frac{1}{n}\sum_{s\in\mathcal{S}}\ln f(\varepsilon_s(\boldsymbol{\theta}))\right)\right|\xrightarrow{p} 0$ for fixed $\boldsymbol{\theta}$. 

To prove this, we want to show that $\mathbb{E}\, |\ln f(\varepsilon_{s,t}(\boldsymbol{\theta}))|<\infty$. Expanding $\ln f(\varepsilon_{s,t}(\boldsymbol{\theta}))$ around $\boldsymbol{\theta}_0$ with respect to $\boldsymbol{\theta}$,
	{\small\setlength{\abovedisplayskip}{3pt}
	\setlength{\belowdisplayskip}{\abovedisplayskip}
	\setlength{\abovedisplayshortskip}{0pt}
	\setlength{\belowdisplayshortskip}{3pt}
	\begin{align*}
	\ln f(\varepsilon_{s,t}(\boldsymbol{\theta})) &= \ln f(\varepsilon_{s,t}(\boldsymbol{\theta}_0))+ \left|\frac{f^{\prime}(\varepsilon_{s,t}(\tilde{\boldsymbol{\theta}}_{n,T}))}{f(\varepsilon_{s,t}(\tilde{\boldsymbol{\theta}}_n))}\frac{\partial \varepsilon_{s,t}(\tilde{\boldsymbol{\theta}}_{n,T})}{\partial \boldsymbol{\theta}^{\prime}}\right|(\boldsymbol{\theta}-\boldsymbol{\theta}_0)\\
	\mathbb{E}\,|\ln f(\varepsilon_{s,t}(\boldsymbol{\theta}))| &\leq \mathbb{E}\,|\ln f(\varepsilon_{s,t}(\boldsymbol{\theta}_0))|+\mathbb{E}\,\left|\frac{f^{\prime}(\varepsilon_{s,t}(\tilde{\boldsymbol{\theta}}_{n,T}))}{f(\varepsilon_{s,t}(\tilde{\boldsymbol{\theta}}_{n,T}))}\frac{\partial \varepsilon_{s,t}(\tilde{\boldsymbol{\theta}}_{n,T})}{\partial \boldsymbol{\theta}^{\prime}}\right||\boldsymbol{\theta}-\boldsymbol{\theta}_0|
	\end{align*}}%
where $\tilde{\boldsymbol{\theta}}_{n,T}$ is between $\boldsymbol{\theta}$ and $\boldsymbol{\theta}_0$. Under the true parameter values, $\varepsilon_{s,t}(\boldsymbol{\theta}_0)$ (denoted as $\varepsilon_{s,t}$ or $\boldsymbol{\varepsilon}_t$ as its vector form in the following) is independent and identically distributed. From Assumption \ref{error-dist}, $\mathbb{E}\, |\ln f(\varepsilon_{s,t})|<\infty$. For $\mathbb{E}\,\left|\frac{f^{\prime}(\varepsilon_{s,t}(\tilde{\boldsymbol{\theta}}_{n,T}))}{f(\varepsilon_{s,t}(\tilde{\boldsymbol{\theta}}_{n,T}))}\frac{\partial \varepsilon_{s,t}(\tilde{\boldsymbol{\theta}}_{n,T})}{\partial \boldsymbol{\theta}^{\prime}}\right|$, $\left|\frac{\partial \varepsilon_{s,t}(\tilde{\boldsymbol{\theta}}_{n,T})}{\partial \boldsymbol{\theta}}\right|$ can be expressed as
{\small\setlength{\abovedisplayskip}{3pt}
	\setlength{\belowdisplayskip}{\abovedisplayskip}
	\setlength{\abovedisplayshortskip}{0pt}
	\setlength{\belowdisplayshortskip}{3pt}
	\begin{align}\label{error-theta}\nonumber
	\left|\frac{\partial \varepsilon_{s,t}(\tilde{\boldsymbol{\theta}}_{n,T})}{\partial \beta}\right| &= \left|x_{s,t}\right|\\
	\left|\frac{\partial \varepsilon_{s,t}(\tilde{\boldsymbol{\theta}}_{n,T})}{\partial \lambda}\right| &= \left|\boldsymbol{F}(x_{s,t}^{\prime}\tilde{\boldsymbol{\gamma}}_{n,T})^{\prime}\right|\leq \boldsymbol{1}_h\\\nonumber
	\left|\frac{\partial \varepsilon_{s,t}(\tilde{\boldsymbol{\theta}}_{n,T})}{\partial \boldsymbol{\gamma}_i}\right| &= \left|\tilde{\lambda}_i\frac{\partial F(x_{s,t}^{\prime}\tilde{\boldsymbol{\gamma}}_i)}{\partial x_{s,t}^{\prime}\boldsymbol{\gamma}_i}x_{s,t}\right|=\left|\tilde{\lambda}_i F(x_{s,t}^{\prime}\tilde{\boldsymbol{\gamma}}_i)(1-F(x_{s,t}^{\prime}\tilde{\boldsymbol{\gamma}}_i))x_{s,t}\right|\\\nonumber
	&\leq \max_{\lambda_{i}\in\boldsymbol{\Theta}}\frac{|\lambda_{i}x_{s,t}|}{4}\\\nonumber
	\left|\frac{\partial \varepsilon_{s,t}(\tilde{\boldsymbol{\theta}}_{n,T})}{\partial \phi_i}\right| &=  \left|\sum_{k=1}^{n}w_{sk}y_{k,t-i}\right| =\left|\Big[W_nA^{-1}(L)(\boldsymbol{g}(X_{t-i},\boldsymbol{\theta}_{0})+\boldsymbol{\varepsilon}_{t-i}(\boldsymbol{\theta}_{0})\Big]_s\right|
	\end{align}}%
where $A^{-1}(L)(\boldsymbol{g}(X_{t-i},\boldsymbol{\theta}_{0})+\boldsymbol{\varepsilon}_{t-i}(\boldsymbol{\theta}_{0})=
 	\sum_{j=0}^{\infty}\Psi_j A_0^{-1}(\boldsymbol{g}(X_{t-i-j},\boldsymbol{\theta}_{0})+\boldsymbol{\varepsilon}_{t-i-j}(\boldsymbol{\theta}_{0}))$. 
Function $g(x_{s,t},\boldsymbol{\theta})=x_{s,t}^{\prime}\beta+\boldsymbol{F}(x_{s,t}^{\prime}\gamma)\lambda$. 
Consider $\boldsymbol{\varepsilon}_{t}(\tilde{\boldsymbol{\theta}}_{n,T})$,
 {\small\setlength{\abovedisplayskip}{3pt}
	\setlength{\belowdisplayskip}{\abovedisplayskip}
	\setlength{\abovedisplayshortskip}{0pt}
	\setlength{\belowdisplayshortskip}{3pt}
	\begin{align*}
	|\boldsymbol{\varepsilon}_{t}(\tilde{\boldsymbol{\theta}}_{n,T})|=&\left|(I_n-\tilde{\phi}_0W_n)Y_t-\sum_{i=1}^{p}\tilde{\phi}_{i}W_nY_{t-i}-\boldsymbol{g}(X_{t},\tilde{\boldsymbol{\theta}}_{n,T})\right|\\
	=&\left|\boldsymbol{\varepsilon}_{t}(\boldsymbol{\theta}_{0})+\sum_{i=0}^{p}(\phi_{i0}-\tilde{\phi}_{i})W_nY_{t-i}+(\boldsymbol{g}(X_{t},\boldsymbol{\theta}_{0})-\boldsymbol{g}(X_{t},\tilde{\boldsymbol{\theta}}_{n,T}))\right|\\
	=&\left|\boldsymbol{\varepsilon}_{t}+\sum_{i=0}^{p}(\phi_{i0}-\tilde{\phi}_{i})W_n\sum_{j=0}^{\infty}\Psi_j A_0^{-1}\boldsymbol{\varepsilon}_{t-i-j}+\sum_{i=0}^{p}(\phi_{i0}-\tilde{\phi}_{i})W_n\sum_{j=0}^{\infty}\Psi_j A_0^{-1}X_{t-i-j}\beta_{0}\right.\\
	&+\left.\sum_{i=0}^{p}(\phi_{i0}-\tilde{\phi}_{i})W_n\sum_{j=0}^{\infty}\Psi_j A_0^{-1}\boldsymbol{F}(X_{t-i-j}\boldsymbol{\gamma}^{\prime}_0)\lambda_0+X_t(\beta_0-\tilde{\beta})+\boldsymbol{F}(X_t\boldsymbol{\gamma}_0^{\prime})\lambda_0-\boldsymbol{F}(X_t\tilde{\boldsymbol{\gamma}}^{\prime})\tilde{\lambda}\right|\\
	<& \left|\boldsymbol{\varepsilon}_{t}+\sum_{i=0}^{p}(\phi_{i0}-\tilde{\phi}_{i})W_n\sum_{j=0}^{\infty}\Psi_j A_0^{-1}\boldsymbol{\varepsilon}_{t-i-j}+\sum_{i=0}^{p}(\phi_{i0}-\tilde{\phi}_{i})W_n\sum_{j=0}^{\infty}\Psi_j A_0^{-1}X_{t-i-j}\beta_{0}\right.\\
	&+\left.\sum_{i=0}^{p}(\phi_{i0}-\tilde{\phi}_{i})W_n\sum_{j=0}^{\infty}\Psi_j A_0^{-1}\boldsymbol{F}(X_{t-i-j}\boldsymbol{\gamma}^{\prime}_0)\lambda_0+X_t(\beta_0-\tilde{\beta})\right|+ ||\lambda_0-\tilde{\lambda}||\cdot\boldsymbol{1}_n
	\end{align*}}%
Denote $P(x^{c})$ is a polynomial about $x$ with highest order $c$. Since we have assumed that $A^{-1}(L)$ existed and the expansion $\sum_{j=0}^{\infty}\Psi_j A_0^{-1}$ is absolutely summable so $W_n\sum_{j=0}^{\infty}\Psi_j A_0^{-1}$ is finite. By Assumption \ref{error-dominance}-\ref{x-moment}, $\left|\frac{f^{\prime}(\varepsilon_{s,t}(\tilde{\boldsymbol{\theta}}_{n,T}))}{f(\varepsilon_{s,t}(\tilde{\boldsymbol{\theta}}_{n,T}))}\right|< a_1+a_2|\varepsilon_{s,t}(\tilde{\boldsymbol{\theta}}_{n,T})|^{c_1}$ and $\mathbb{E}\left|\frac{f^{\prime}(\varepsilon_{s,t})}{f(\varepsilon_{s,t})}\right|, \mathbb{E}\left|\frac{f^{\prime}(\varepsilon_{s,t})}{f(\varepsilon_{s,t})}\right|^2$ are dominated by $a_1+a_2|\varepsilon_{s,t}|^{c_1}$, $\mathbb{E}|\varepsilon_{s,t}|^{c_1}<\infty, \,\mathbb{E}|x_{s,t}|^{c_1}<\infty$. Let $c^{\ast}=\max(1,c_1)$, then,
{\small\setlength{\abovedisplayskip}{3pt}
	\setlength{\belowdisplayskip}{\abovedisplayskip}
	\setlength{\abovedisplayshortskip}{0pt}
	\setlength{\belowdisplayshortskip}{3pt}
	\begin{align*}
	\mathbb{E}\left|\frac{f^{\prime}(\varepsilon_{s,t}(\tilde{\boldsymbol{\theta}}_{n,T}))}{f(\varepsilon_{s,t}(\tilde{\boldsymbol{\theta}}_{n,T}))}\right|^2< 	P(\mathbb{E}\,\left|\varepsilon_{s,t}\right|^{c^{\ast}})+ P(\mathbb{E}\,\left|x_{s,t}\right|^{c^{\ast}})+Constant<\infty
	\end{align*}}%
So also $\mathbb{E}\left|\frac{f^{\prime}(\varepsilon_{s,t}(\tilde{\boldsymbol{\theta}}_{n,T}))}{f(\varepsilon_{s,t}(\tilde{\boldsymbol{\theta}}_{n,T}))}\right|<\infty$. With Cauchy–Schwarz inequality \cite{steele2004cauchy} and the finite second moment of $x_{s,t}$, we can have,
{\small\setlength{\abovedisplayskip}{3pt}
	\setlength{\belowdisplayskip}{\abovedisplayskip}
	\setlength{\abovedisplayshortskip}{0pt}
	\setlength{\belowdisplayshortskip}{3pt}
	\begin{align}\label{log-derivative-bound}
	\mathbb{E}\left|\frac{f^{\prime}(\varepsilon_{s,t}(\tilde{\boldsymbol{\theta}}_{n,T}))}{f(\varepsilon_{s,t}(\tilde{\boldsymbol{\theta}}_{n,T}))}\frac{\partial \varepsilon_{s,t}(\tilde{\boldsymbol{\theta}}_{n,T})}{\partial \beta}\right| &= \mathbb{E}\left|\frac{f^{\prime}(\varepsilon_{s,t}(\tilde{\boldsymbol{\theta}}_{n,T}))}{f(\varepsilon_{s,t}(\tilde{\boldsymbol{\theta}}_{n,T}))}x_{s,t}\right|<\left(\mathbb{E}\left|\frac{f^{\prime}(\varepsilon_{s}(\tilde{\boldsymbol{\theta}}_{n,T}))}{f(\varepsilon_{s}(\tilde{\boldsymbol{\theta}}_{n,T}))}\right|^2\mathbb{E}\left|x_{s,t}\right|^2\right)^{1/2}<\infty\\
	\mathbb{E}\left|\frac{f^{\prime}(\varepsilon_{s,t}(\tilde{\boldsymbol{\theta}}_{n,T}))}{f(\varepsilon_{s,t}(\tilde{\boldsymbol{\theta}}_{n,T}))}\frac{\partial \varepsilon_{s,t}(\tilde{\boldsymbol{\theta}}_{n,T})}{\partial \lambda}\right| &= \mathbb{E}\left|\frac{f^{\prime}(\varepsilon_{s,t}(\tilde{\boldsymbol{\theta}}_{n,T}))}{f(\varepsilon_{s,t}(\tilde{\boldsymbol{\theta}}_{n,T}))}\boldsymbol{F}(x_{s,t}^{\prime}\tilde{\boldsymbol{\gamma}})^{\prime}\right|\leq \mathbb{E}\left|\frac{f^{\prime}(\varepsilon_{s,t}(\tilde{\boldsymbol{\theta}}_{n,T}))}{f(\varepsilon_{s,t}(\tilde{\boldsymbol{\theta}}_{n,T}))}\boldsymbol{1}_h\right|<\infty\\
	\mathbb{E}\left|\frac{f^{\prime}(\varepsilon_{s,t}(\tilde{\boldsymbol{\theta}}_{n,T}))}{f(\varepsilon_{s,t}(\tilde{\boldsymbol{\theta}}_{n,T}))}\frac{\partial \varepsilon_{s,t}(\tilde{\boldsymbol{\theta}}_{n,T})}{\partial \boldsymbol{\gamma}_i}\right| &\leq \mathbb{E}\left|\frac{f^{\prime}(\varepsilon_{s,t}(\tilde{\boldsymbol{\theta}}_{n,T}))}{f(\varepsilon_{s,t}(\tilde{\boldsymbol{\theta}}_{n,T}))}\tilde{\lambda}_ix_{s,t}\right|< \infty\\
	\mathbb{E}\left|\frac{f^{\prime}(\varepsilon_{s,t}(\tilde{\boldsymbol{\theta}}_{n,T}))}{f(\varepsilon_{s,t}(\tilde{\boldsymbol{\theta}}_{n,T}))}\frac{\partial \varepsilon_{s,t}(\tilde{\boldsymbol{\theta}}_{n,T})}{\partial \phi_i}\right| &= \mathbb{E}\left|\frac{f^{\prime}(\varepsilon_{s,t}(\tilde{\boldsymbol{\theta}}_{n,T}))}{f(\varepsilon_{s,t}(\tilde{\boldsymbol{\theta}}_{n,T}))}\Big[W_nA^{-1}(L)(\boldsymbol{g}(X_{t-i},\boldsymbol{\theta}_{0})+\boldsymbol{\varepsilon}_{t-i}(\boldsymbol{\theta}_{0})\Big]_s\right|\\\label{log-derivative-bound-phi-e}
	&< \mathbb{E}\left|\frac{f^{\prime}(\varepsilon_{s,t}(\tilde{\boldsymbol{\theta}}_{n,T}))}{f(\varepsilon_{s,t}(\tilde{\boldsymbol{\theta}}_{n,T}))}\Big[W_nA^{-1}(L)\boldsymbol{\varepsilon}_{t-i}(\boldsymbol{\theta}_{0})\Big]_s\right|\\\label{log-derivative-bound-phi-x}
	&+k_0\mathbb{E}\left|\Big[W_nA^{-1}(L)\boldsymbol{g}(X_{t-i},\boldsymbol{\theta}_{0})\Big]_s\right|
	\, i=0,\ldots,p
	\end{align}}%
Because $W_nA^{-1}(L)$ is well defined and $X_t$ is stationary with finite second moment, so component (\ref{log-derivative-bound-phi-x}) is finite.
(\ref{log-derivative-bound-phi-e}) is dominated by $P(\mathbb{E}\,|\varepsilon_{s,t}|^{c^{\ast}+1})$ so with the dominance assumption, (\ref{log-derivative-bound-phi-e}) is finite.
Hence, with (\ref{log-derivative-bound})-(\ref{log-derivative-bound-phi-x}) finite, $\mathbb{E}\,|\ln f(\varepsilon_{s,t}(\boldsymbol{\theta}_0))|<\infty$, we can conclude that
$\mathbb{E}\,|\ln f(\varepsilon_{s,t}(\boldsymbol{\theta}))| < \infty$. Then by ergodic theorem \cite{birkhoff1931proof},
{\small\setlength{\abovedisplayskip}{3pt}
	\setlength{\belowdisplayskip}{\abovedisplayskip}
	\setlength{\abovedisplayshortskip}{0pt}
	\setlength{\belowdisplayshortskip}{3pt}
	\begin{gather*}
	\left|\frac{1}{nT}\sum_{t=1}^{T}\left(\sum_{s\in\mathcal{S}}\ln f(\varepsilon_{s,t}(\boldsymbol{\theta}))-\mathbb{E}\,\frac{1}{n}\sum_{s\in\mathcal{S}}\ln f(\varepsilon_{s,t}(\boldsymbol{\theta}))\right)\right | \xrightarrow{p} 0,\quad n_1,n_2,T\rightarrow\infty
	\end{gather*}}%

To complete the proof of uniform convergence,  we also need to show $\frac{1}{nT}\sum_{t=1}^{T}\sum_{s\in\mathcal{S}}\ln f(\varepsilon_{s,t}(\boldsymbol{\theta}))$ is equicontinuous for $\boldsymbol{\theta}\in\boldsymbol{\Theta}$, i.e., for all $\boldsymbol{\theta}_1, \boldsymbol{\theta}_2\in \boldsymbol{\Theta}$,
{\small\setlength{\abovedisplayskip}{3pt}
	\setlength{\belowdisplayskip}{\abovedisplayskip}
	\setlength{\abovedisplayshortskip}{0pt}
	\setlength{\belowdisplayshortskip}{3pt}
	\begin{align}\label{equi}
	\frac{1}{nT}\left|\sum_{t=1}^{T}\sum_{s\in\mathcal{S}}\Big(\ln f(\varepsilon_{s,t}(\boldsymbol{\theta}_1))-\ln f(\varepsilon_{s,t}(\boldsymbol{\theta}_2))\Big)\right| \leq ||\boldsymbol{\theta}_1-\boldsymbol{\theta}_2||O_p(1)
	\end{align}}%
Applying the mean value theorem to the left side in (\ref{equi}):
{\small\setlength{\abovedisplayskip}{3pt}
	\setlength{\belowdisplayskip}{\abovedisplayskip}
	\setlength{\abovedisplayshortskip}{0pt}
	\setlength{\belowdisplayshortskip}{3pt}
	\begin{align*}
	\frac{1}{nT}\left|\sum_{t=1}^{T}\sum_{s\in\mathcal{S}}\Big(\ln f(\varepsilon_{s,t}(\boldsymbol{\theta}_1))-\ln f(\varepsilon_{s,t}(\boldsymbol{\theta}_2))\Big)\right| &\leq \frac{1}{nT}\left|\sum_{t=1}^{T}\sum_{s\in\mathcal{S}}\frac{\partial \ln f(\varepsilon_{s,t}(\tilde{\boldsymbol{\theta}}_{n,T}))}{\partial \boldsymbol{\theta}^{\prime}}\right|||\boldsymbol{\theta}_1-\boldsymbol{\theta}_2||\\
	&=\frac{1}{nT}\left|\sum_{t=1}^{T}\sum_{s\in\mathcal{S}}\frac{ f^{\prime}(\varepsilon_{s,t}(\tilde{\boldsymbol{\theta}}_{n,T}))}{ f(\varepsilon_{s,t}(\tilde{\boldsymbol{\theta}}_{n,T}))}\frac{\partial \varepsilon_{s,t}(\tilde{\boldsymbol{\theta}}_{n,T})}{\partial \boldsymbol{\theta}^{\prime}}\right| ||\boldsymbol{\theta}_1-\boldsymbol{\theta}_2||
	\end{align*}}%
where $\tilde{\boldsymbol{\theta}}_{n,T}$ is some value between $\boldsymbol{\theta}_1$ and $\boldsymbol{\theta}_2$. 
Since $\boldsymbol{\theta}$ is in a compact set $\boldsymbol{\Theta}$, we show in (\ref{res-expansion}) that, for all $s,t$, $\varepsilon_{s,t}(\boldsymbol{\theta})$ is bounded by some function of $Z_t$ not depending on $\boldsymbol{\theta}$.
{\small\setlength{\abovedisplayskip}{3pt}
	\setlength{\belowdisplayskip}{\abovedisplayskip}
	\setlength{\abovedisplayshortskip}{0pt}
	\setlength{\belowdisplayshortskip}{3pt}
	\begin{align}\label{res-expansion}\nonumber
	|\boldsymbol{\varepsilon}_t(\boldsymbol{\theta})|&=\left|Y_t-\phi_0 W_nY_t-\sum_{k=1}^{p}\phi_kW_nY_{t-k}-X_t\beta-\boldsymbol{F}(X_t\boldsymbol{\gamma}^{\prime})\lambda\right|\\
	&\leq \left|(I_n-\phi_0 W_n)Y_t\right|+|\sum_{k=1}^{p}\phi_kW_nY_{t-k}|+\left|X_n\beta\right| +\left|\boldsymbol{F}(X_n\boldsymbol{\gamma}^{\prime})\lambda\right|\\\nonumber
	&\leq (I_n+\max_{\phi_0\in\boldsymbol{\Theta}}|\phi_0|W_n)|Y_t|+\sum_{k=1}^{p}\max_{\phi_i\in\boldsymbol{\Theta}}W_n|\phi_iY_{t-k}|+|X_n|\max_{\beta\in\boldsymbol{\Theta}}|\beta|+\max_{\lambda\in\boldsymbol{\Theta}}||\lambda||\boldsymbol{1}_n
	\end{align}}%
Similarly, referring to (\ref{error-theta}), it is easy to show that $\left|\frac{\partial\varepsilon_{s,t}(\boldsymbol{\theta})}{\partial \boldsymbol{\theta}}\right|$ is bounded by some function about $Y_t$ and $X_t$.
Therefore, due to the dominance of $\left|\frac{f^{\prime}(s)}{f(s)}\right|$ (see Assumption \ref{error-dominance}) and stationarity of $X_t,Y_t$, for $\tilde{\boldsymbol{\theta}}_{n,T}$ between $\boldsymbol{\theta}_1$ and $\boldsymbol{\theta}_2$, there exists a constant $M$ such that
{\small\setlength{\abovedisplayskip}{3pt}
	\setlength{\belowdisplayskip}{\abovedisplayskip}
	\setlength{\abovedisplayshortskip}{0pt}
	\setlength{\belowdisplayshortskip}{3pt}
	\begin{align}\label{loglikelihood-theta-expectation}
	\frac{1}{nT}\left|\sum_{t=1}^{T}\sum_{s\in\mathcal{S}}\frac{ f^{\prime}(\varepsilon_{s,t}(\tilde{\boldsymbol{\theta}}_{n,T}))}{ f(\varepsilon_{s,t}(\tilde{\boldsymbol{\theta}}_{n,T}))}\frac{\partial \varepsilon_{s,t}(\tilde{\boldsymbol{\theta}}_{n,T})}{\partial \boldsymbol{\theta}^{\prime}}\right|\leq M\quad \text{for }n_1,n_2, T\rightarrow\infty
	\end{align}}%
Hence, for $\boldsymbol{\theta}_1,\boldsymbol{\theta}_2 \in \boldsymbol{\Theta}$
{\small\setlength{\abovedisplayskip}{3pt}
	\setlength{\belowdisplayskip}{\abovedisplayskip}
	\setlength{\abovedisplayshortskip}{0pt}
	\setlength{\belowdisplayshortskip}{3pt}
	\begin{align*}
	\frac{1}{nT}\left|\sum_{t=1}^{T}\sum_{s\in\mathcal{S}}\Big(\ln f(\varepsilon_{s,t}(\boldsymbol{\theta}_1))-\ln f(\varepsilon_{s,t}(\boldsymbol{\theta}_2))\Big)\right| = ||\boldsymbol{\theta}_1-\boldsymbol{\theta}_2|| O_p(1)
	\end{align*}}%
So $\frac{1}{nT}\sum_{t=1}^{T}\sum_{s\in\mathcal{S}}\ln f(\varepsilon_{s,t}(\boldsymbol{\theta}))$ is equicontinuous for $\boldsymbol{\theta}\in\boldsymbol{\Theta}$.
With the pointwise convergence and equicontinuity, we can conclude the uniform convergence in (\ref{uni-cvg-eqn-2}) and furthermore (\ref{uni-cvg-eqn}) follows. 
\end{proof}

Similar to Chapter 1, we now give a formal statement of the consistency results.
\begin{thm}
	Given Assumptions \ref{compact-space}-\ref{x-moment}, $\hat{\boldsymbol{\theta}}_{n,T}\xrightarrow{p}\boldsymbol{\theta}_0$ as $n,T\rightarrow\infty$.
\end{thm}
\begin{proof}
	Similar to the proof by Lung-fei Lee \cite{lee2004asymptotic}, we need to show the stochastic equicontinuity of $\frac{1}{n}\ln |A_0|$ to have the uniform convergence of the log likelihood function $\mathcal{L}_{n,T}(\boldsymbol{\theta})$. Applying the mean value theorem, 
	{\small\setlength{\abovedisplayskip}{3pt}
		\setlength{\belowdisplayskip}{\abovedisplayskip}
		\setlength{\abovedisplayshortskip}{0pt}
		\setlength{\belowdisplayshortskip}{3pt}
		\begin{align*}
		\left|\frac{1}{n}(\ln |I_n-\phi_{0}^{\dagger}W_n|-\ln |I_n-\phi_{0}^{\ddagger}W_n|)\right| = \left|(\phi_{0}^{\dagger}-\phi_{0}^{\ddagger})\frac{1}{n}tr(W_n(I_n-\phi^{\ast}_{0_{n,T}}W_n)^{-1})\right|
		\end{align*}} %
	where $\phi_{0_{n,T}}^{\ast}$ is between $\phi_{0}^{\dagger}$ and $\phi_{0}^{\ddagger}$. By Assumption \ref{phi0-range} and \ref{ub-weight-matrix}, $\sup_{\phi_0\in\boldsymbol{\Theta}}|\phi_0|<1$, $W_n$ is bounded in both rows and column sums uniformly and using (\ref{matrix-decomposition}), 
	{\small\setlength{\abovedisplayskip}{3pt}
		\setlength{\belowdisplayskip}{\abovedisplayskip}
		\setlength{\abovedisplayshortskip}{0pt}
		\setlength{\belowdisplayshortskip}{3pt}
		\begin{align*}
		\left|\frac{1}{n}tr(W_n(I_n-\phi^{\ast}_{0_{n,T}}W_n)^{-1})\right| = \left|\frac{1}{n}\sum_{i=1}^{n}\frac{\tau_i}{1-\phi^{\ast}_{0_{n,T}}\tau_i}\right|\leq C_{1}
		\end{align*}} %
	where $C_{1}$ is a constant not depending on $n$. So $		\left|\frac{1}{n}(\ln |I_n-\phi_{0}^{\dagger}W_n|-\ln |I_n-\phi_{0}^{\ddagger}W_n|)\right| \leq C_1|\phi_{0}^{\dagger}-\phi_{0}^{\ddagger}|$ and with Lemma \ref{uni-cvg} we can conclude the uniform convergence that
	{\small\setlength{\abovedisplayskip}{3pt}
		\setlength{\belowdisplayskip}{\abovedisplayskip}
		\setlength{\abovedisplayshortskip}{0pt}
		\setlength{\belowdisplayshortskip}{3pt}
		\begin{align*}
		\sup_{\boldsymbol{\theta}\in\boldsymbol{\Theta}}\left|\frac{1}{nT}\mathcal{L}_{n,T}(\boldsymbol{\theta})-\mathbb{E}\, \frac{1}{nT}\mathcal{L}_{n,T}(\boldsymbol{\theta})\right| \xrightarrow{p} 0.
		\end{align*}} %
	With the assumptions \ref{compact-space}-\ref{error-dominance}, the parameter space $\boldsymbol{\Theta}$ is compact; $\frac{1}{nT}\mathcal{L}_{n,T}(\boldsymbol{\theta})$ is continuous in $\boldsymbol{\theta}\in \boldsymbol{\Theta}$ and is a measurable of $Y_t, X_t, t=1,\ldots,T$ for all $\boldsymbol{\theta}\in \boldsymbol{\Theta}$. $\mathbb{E}\,\frac{1}{nT}\mathcal{L}_{n,T}(\boldsymbol{\theta})$ is continuous on $\boldsymbol{\Theta}$ and by Lemma \ref{unique-max}, $\mathbb{E}\,\frac{1}{nT}\mathcal{L}_{n,T}(\boldsymbol{\theta})$ has a unique maximum at $\boldsymbol{\theta}_0$. Referring to Theorem 3.5 in White \cite{white1994parametric} with the uniform convergence in (\ref{uni-cvg-eqn}), we can conclude that $\hat{\boldsymbol{\theta}}_{n,T}\xrightarrow{p}\boldsymbol{\theta}_0$ as $n,T\rightarrow\infty$.
\end{proof}
\subsection{Asymptotic Distribution}
\begin{ass}\label{Amatrix-limit}
	The limit $A(\boldsymbol{\theta}_0)=-\lim_{n,T\rightarrow\infty}\mathbb{E}\,\frac{1}{nT}\frac{\partial^2\mathcal{L}_{n,T}(\boldsymbol{\theta}_0)}{\partial\boldsymbol{\theta}\partial\boldsymbol{\theta}^{\prime}}$ is nonsingular.
\end{ass}
\begin{ass}\label{Bmatrix-limit}
	The limit $B(\boldsymbol{\theta}_0)=\lim_{n,T\rightarrow\infty}\mathbb{E}\,\frac{1}{nT}\frac{\partial\mathcal{L}_{n,T}(\boldsymbol{\theta}_0)}{\partial\boldsymbol{\theta}}\frac{\partial\mathcal{L}_{n,T}(\boldsymbol{\theta}_0)}{\partial\boldsymbol{\theta}^{\prime}}$ is nonsingular.
\end{ass}
These assumptions are to guarantee the existence of the covariance matrix of the limiting distribution of parameters in a PSTAR(p)-ANN model. 
We now give the asymptotic distribution of the maximum likelihood estimator $\hat{\boldsymbol{\theta}}_{n,T}$. 
\begin{thm}\label{asythm}
	Under Assumptions \ref{compact-space}-\ref{Bmatrix-limit},
	{\small\setlength{\abovedisplayskip}{3pt}
		\setlength{\belowdisplayskip}{\abovedisplayskip}
		\setlength{\abovedisplayshortskip}{0pt}
		\setlength{\belowdisplayshortskip}{3pt}
		\begin{equation}\label{asymptotic}
		\sqrt{nT} (\hat{\boldsymbol{\theta}}_{n,T}-\boldsymbol{\theta}_0) \xrightarrow{d} N (\boldsymbol{0}, \boldsymbol{\Omega}_0)
		\end{equation}} %
	where $\boldsymbol{\Omega}_0= A(\boldsymbol{\theta}_0)^{-1}B(\boldsymbol{\theta}_0)A(\boldsymbol{\theta}_0)^{-1} = A(\boldsymbol{\theta}_0)^{-1}$
\end{thm}
\begin{proof}
	Since $\hat{\boldsymbol{\theta}}_{n,T}$ maximizes $\mathcal{L}_{n,T}(\boldsymbol{\theta})$, $\frac{\partial \mathcal{L}_{n,T}(\hat{\boldsymbol{\theta}}_{n,T})}{\partial \boldsymbol{\theta}} = 0$. By the mean value theorem, expand $\frac{\partial \mathcal{L}_{n,T}(\hat{\boldsymbol{\theta}}_{n,T})}{\partial \boldsymbol{\theta}}$ around $\boldsymbol{\theta}_0$ with respect to $\boldsymbol{\theta}$,
{\small\setlength{\abovedisplayskip}{3pt}
	\setlength{\belowdisplayskip}{\abovedisplayskip}
	\setlength{\abovedisplayshortskip}{0pt}
	\setlength{\belowdisplayshortskip}{3pt}
	\begin{gather*}
	\frac{\partial \mathcal{L}_{n,T}(\hat{\boldsymbol{\theta}}_{n,T})}{\partial \boldsymbol{\theta}}  = \frac{\partial \mathcal{L}_{n,T}(\boldsymbol{\theta}_0)}{\partial \boldsymbol{\theta}}+\frac{\partial^2 \mathcal{L}_{n,T}(\tilde{\boldsymbol{\theta}}_{n,T})}{\partial \boldsymbol{\theta}\partial \boldsymbol{\theta}^{\prime}}(\hat{\boldsymbol{\theta}}_{n,T}-\boldsymbol{\theta}_0)\\
	0 = \frac{\partial \mathcal{L}_{n,T}(\boldsymbol{\theta}_0)}{\partial \boldsymbol{\theta}}+\frac{\partial^2 \mathcal{L}_{n,T}(\tilde{\boldsymbol{\theta}}_{n,T})}{\partial \boldsymbol{\theta}\partial \boldsymbol{\theta}^{\prime}}(\hat{\boldsymbol{\theta}}_{n,T}-\boldsymbol{\theta}_0)
	\end{gather*}} %
where $\tilde{\boldsymbol{\theta}}_{n,T}$ is between $\hat{\boldsymbol{\theta}}_{n,T}$ and $\boldsymbol{\theta}_0$. Therefore, we can have the following equation:
{\small\setlength{\abovedisplayskip}{3pt}
		\setlength{\belowdisplayskip}{\abovedisplayskip}
		\setlength{\abovedisplayshortskip}{0pt}
		\setlength{\belowdisplayshortskip}{3pt}
		\begin{equation}\label{taylor}
		\sqrt{nT} (\hat{\boldsymbol{\theta}}_{n,T}-\boldsymbol{\theta}_0) =\left[-\frac{1}{nT}\frac{\partial^2\mathcal{L}_{n,T}(\tilde{\boldsymbol{\theta}}_{n,T})}{\partial\boldsymbol{\theta}\partial\boldsymbol{\theta}^{\prime}}\right]^{-1}\frac{1}{\sqrt{nT}}\frac{\partial\mathcal{L}_{n,T}(\boldsymbol{\theta}_0)}{\partial\boldsymbol{\theta}}
		\end{equation}} %
From (\ref{error-theta}), denote $\frac{\boldsymbol{f}^{\prime}(\varepsilon_{t}\boldsymbol{\theta})}{\boldsymbol{f}(\varepsilon_{t}(\boldsymbol{\theta}))}$ as $V_t(\boldsymbol{\theta})\in\mathbb{R}^n$ and $\frac{\boldsymbol{f}^{\prime}(\varepsilon_{t}\boldsymbol{\theta}_0)}{\boldsymbol{f}(\varepsilon_{t}(\boldsymbol{\theta}_0))}=V_t$.

Recall that $Z_t=(W_nY_t, W_nY_{t-1},\ldots, W_nY_{t-p}, X_t)$ so the first order derivatives can be expressed as
	 	{\small\setlength{\abovedisplayskip}{3pt}
	 	\setlength{\belowdisplayskip}{\abovedisplayskip}
	 	\setlength{\abovedisplayshortskip}{0pt}
	 	\setlength{\belowdisplayshortskip}{3pt}
	 	\begin{gather}\label{first-order-derivative}
	 	\frac{1}{\sqrt{nT}}\frac{\partial \mathcal{L}_{n,T}(\boldsymbol{\theta})}{\partial \boldsymbol{\theta}}=
	 	\begin{pmatrix*}[l]
	 	-\frac{1}{\sqrt{nT}}\sum_{t=1}^{T}\big((W_nY_t)^{\prime}V_t(\boldsymbol{\theta})+tr(W_nA_0^{-1})\big)\\
	 	-\frac{1}{\sqrt{nT}}\sum_{t=1}^{T}Z_t^{\prime}V_t(\boldsymbol{\theta})\\
	 	-\frac{1}{\sqrt{nT}}\sum_{t=1}^{T}\big(\boldsymbol{F}(X_t\boldsymbol{\gamma}^{\prime})\big)^{\prime} V_t(\boldsymbol{\theta})\\
	 	-\frac{\lambda_1}{\sqrt{nT}}\sum_{t=1}^{T}X_t^{\prime}\big(\boldsymbol{F}^{\prime}(X_t\boldsymbol{\gamma}_1)\circ V_t(\boldsymbol{\theta})\big)\\
	 	\hspace{9em}\vdots\\
	 	-\frac{\lambda_h}{\sqrt{nT}}\sum_{t=1}^{T}X_t^{\prime}\big(\boldsymbol{F}^{\prime}(X_t\boldsymbol{\gamma}_h)\circ V_t(\boldsymbol{\theta})\big)
	 	\end{pmatrix*}
	 	\end{gather}}%
By Lemma \ref{unique-max}, the true parameter values maximize $\frac{1}{nT}\mathbb{E}\,\mathcal{L}_{n,T}(\boldsymbol{\theta})$, so $\frac{1}{nT}\frac{\partial \mathbb{E}\,\mathcal{L}_{n,T}(\boldsymbol{\theta}_0)}{\partial \boldsymbol{\theta}}=\boldsymbol{0}$. In (\ref{log-derivative-bound})-(\ref{log-derivative-bound-phi-x}) and (\ref{res-expansion}), we showed that $\mathbb{E}\left|\frac{\partial \ln f(\varepsilon_{s,t}(\boldsymbol{\theta}))}{\partial \boldsymbol{\theta}}\right|$ is dominated by some function not related to $\boldsymbol{\theta}$ and (\ref{loglikelihood-theta-expectation}) indicates that $\mathbb{E}\left|\frac{\partial \ln f(\varepsilon_{s,t}(\boldsymbol{\theta}))}{\partial \boldsymbol{\theta}}\right|$ is bounded for interior units in $\mathcal{S}$. Hence,$\mathbb{E}\,\frac{\partial \ln f(\varepsilon_{s,t}(\boldsymbol{\theta}))}{\partial \boldsymbol{\theta}} = \frac{\partial }{\partial \boldsymbol{\theta}}\mathbb{E}\ln f(\varepsilon_{s,t}(\boldsymbol{\theta}))$, it follows that, with $\frac{1}{nT}\mathcal{L}_{n,T}(\boldsymbol{\theta})=\frac{1}{n}\ln|A_0|+\frac{1}{nT}\sum_{s=1}^{n}\sum_{t=1}^{T}\ln f(\varepsilon_{s,t}(\boldsymbol{\theta}))$, we can have,
{\small\setlength{\abovedisplayskip}{3pt}
	\setlength{\belowdisplayskip}{\abovedisplayskip}
	\setlength{\abovedisplayshortskip}{0pt}
	\setlength{\belowdisplayshortskip}{3pt}
	\begin{gather*}	
	\frac{1}{nT}\frac{\partial \mathbb{E}\,\mathcal{L}_{n,T}(\boldsymbol{\theta}_0)}{\partial \boldsymbol{\theta}} = \frac{1}{nT}\mathbb{E}\,\frac{\partial \mathcal{L}_{n,T}(\boldsymbol{\theta}_0)}{\partial \boldsymbol{\theta}}=\boldsymbol{0}
	\end{gather*}} %
Therefore, with Assumption \ref{Bmatrix-limit},
 	{\small\setlength{\abovedisplayskip}{3pt}
 		\setlength{\belowdisplayskip}{\abovedisplayskip}
 		\setlength{\abovedisplayshortskip}{0pt}
 		\setlength{\belowdisplayshortskip}{3pt}
 		\begin{gather*}
 			\text{Var}(\frac{1}{\sqrt{nT}}\frac{\partial \mathcal{L}_{n,T}(\boldsymbol{\theta}_0)}{\partial \boldsymbol{\theta}})=-\mathbb{E}\frac{1}{nT}\frac{\partial^2 \mathcal{L}_{n,T}(\boldsymbol{\theta}_0)}{\partial\boldsymbol{\theta}\partial\boldsymbol{\theta}^{\prime}} =\mathbb{E}\left(\frac{1}{nT}\frac{\partial\mathcal{L}_{n,T}(\boldsymbol{\theta}_0)}{\partial\boldsymbol{\theta}}\frac{\partial\mathcal{L}_{n,T}(\boldsymbol{\theta}_0)}{\partial\boldsymbol{\theta}^{\prime}}\right)\rightarrow B(\boldsymbol{\theta}_0)
 		\end{gather*}}%
And under this $A(\boldsymbol{\theta}_0) = B(\boldsymbol{\theta}_0)$. Since
$\frac{\partial \mathcal{L}_{n,T}(\boldsymbol{\theta}_0)}{\partial \boldsymbol{\theta}}$ is the sum of $T$ identical and ergodic random variables,
by the central limit theorem for stationary ergodic processes \cite{gordin1969clt}, the limiting distribution of $\frac{1}{\sqrt{nT}}\frac{\partial \mathcal{L}_{n,T}(\boldsymbol{\theta}_0)}{\partial \boldsymbol{\theta}}$ is $N(\boldsymbol{0}, B(\boldsymbol{\theta}_0))$.
 		
Next we would like to show that $\frac{1}{nT}\frac{\partial^2\mathcal{L}_{n,T}(\tilde{\boldsymbol{\theta}}_{n,T})}{\partial\boldsymbol{\theta}\partial\boldsymbol{\theta}^{\prime}}-\frac{1}{nT}\frac{\partial^2\mathcal{L}_{n,T}(\boldsymbol{\theta}_0)}{\partial\boldsymbol{\theta}\partial\boldsymbol{\theta}^{\prime}}\xrightarrow{p}0$. Following the results in (\ref{first-order-derivative}), define $U_t(\boldsymbol{\theta})=\frac{\boldsymbol{f}^{\prime\prime}(\varepsilon_t(\boldsymbol{\theta}))}{\boldsymbol{f}(\varepsilon_t(\boldsymbol{\theta}))}-\frac{\boldsymbol{f}^{\prime2}(\varepsilon_t(\boldsymbol{\theta}))}{\boldsymbol{f}^2(\varepsilon_t(\boldsymbol{\theta}))}\in\mathbb{R}^{n}$, and write $U_t = U_t(\boldsymbol{\theta}_0)$ so the second order derivatives are given below $-\frac{1}{nT}\frac{\partial^2\mathcal{L}_{n,T}(\boldsymbol{\theta})}{\partial\boldsymbol{\theta}\partial\boldsymbol{\theta}^{\prime}}=$
 {\small\setlength{\abovedisplayskip}{3pt}
 	\setlength{\belowdisplayskip}{\abovedisplayskip}
 	\setlength{\abovedisplayshortskip}{0pt}
 	\setlength{\belowdisplayshortskip}{3pt}		
 	\begin{gather}\label{second-order-derivative}
 	\frac{1}{nT}\sum_{t=1}^{T}
 		\begin{psmallmatrix*}[l]
 		G_{0,t}(\boldsymbol{\theta})&(W_nY_t)^{\prime}G_{1,t}(\boldsymbol{\theta})&(W_nY_t)^{\prime}G_{2,t}(\boldsymbol{\theta})&(W_nY_t)^{\prime}H_{1,t}(\boldsymbol{\theta})&\cdots&(W_nY_t)^{\prime}H_{h,t}(\boldsymbol{\theta})\\
 		G_{1,t}^{\prime}(\boldsymbol{\theta})W_nY_t&Z_t^{\prime}G_{1,t}(\boldsymbol{\theta})&Z_t^{\prime}G_{2,t}(\boldsymbol{\theta})&Z_t^{\prime}H_{1,t}(\boldsymbol{\theta})&\cdots&Z_t^{\prime}H_{h,t}(\boldsymbol{\theta})\\
 		G_{2,t}^{\prime}(\boldsymbol{\theta})W_nY_t&G_{2,t}^{\prime}(\boldsymbol{\theta})Z_t&\boldsymbol{F}(X_t\boldsymbol{\gamma}^{\prime})^{\prime}G_{2,t}(\boldsymbol{\theta})&\boldsymbol{F}(X_t\boldsymbol{\gamma}^{\prime})^{\prime}H_{1,t}(\boldsymbol{\theta})&\cdots&\boldsymbol{F}(X_t\boldsymbol{\gamma}^{\prime})^{\prime}H_{h,t}(\boldsymbol{\theta})\\
 		&&&+K_{1,t}(\boldsymbol{\theta})&\cdots&+K_{h,t}(\boldsymbol{\theta})\\
 		H_{1,t}^{\prime}(\boldsymbol{\theta})W_nY_t&H_{1,t}^{\prime}(\boldsymbol{\theta})Z_t&H_{1,t}^{\prime}(\boldsymbol{\theta})\boldsymbol{F}(X_t\boldsymbol{\gamma}^{\prime})\\
 		&&+K_{1,t}(\boldsymbol{\theta})^{\prime}&&\\
 		\vdots&\vdots&\vdots&&J(\boldsymbol{\theta})&\\
 		H_{h,t}^{\prime}(\boldsymbol{\theta})W_nY_t&H_{h,t}^{\prime}(\boldsymbol{\theta})Z_t&H_{h,t}^{\prime}(\boldsymbol{\theta})\boldsymbol{F}(X_t\boldsymbol{\gamma}^{\prime})\\
 		&&+K_{h,t}(\boldsymbol{\theta})^{\prime}&&\\
 		\end{psmallmatrix*}
 		\end{gather}}%
  {\small\setlength{\abovedisplayskip}{3pt}
 	\setlength{\belowdisplayskip}{\abovedisplayskip}
 	\setlength{\abovedisplayshortskip}{0pt}
 	\setlength{\belowdisplayshortskip}{3pt}		
 	\begin{align*}	
 		J_{ij,t}(\boldsymbol{\theta}) &=\left\{\begin{array}{ll}
 			\lambda_{i}X_t^{\prime}(\boldsymbol{F}^{\prime\prime}(X_t\boldsymbol{\gamma}_{i})\circ V_t(\boldsymbol{\theta})\circ  X_t)+\lambda_{i}X_t^{\prime}(\boldsymbol{F}^{\prime}(X_t\boldsymbol{\gamma}_{i})\circ H_{i,t})&i=j\\
 			\lambda_{i}(\boldsymbol{F}^{\prime}(X_t\boldsymbol{\gamma}_{i})\circ H_{j,t})^{\prime}X_t&i>j \quad i,j ={1,2,\ldots,h}\\
 			\lambda_{i}X_t^{\prime}(\boldsymbol{F}^{\prime}(X_t\boldsymbol{\gamma}_{i})\circ H_{j,t})&i< j
 			\end{array}\right.
 		\\\nonumber
 		G_{0,t}(\boldsymbol{\theta}) &=\left(-W_nY_t\circ W_nY_{t})^{\prime}U_t(\boldsymbol{\theta})+tr((W_nA_0^{-1})^2\right)\\\nonumber
 		G_{1,t}(\boldsymbol{\theta})&=-U_t(\boldsymbol{\theta})\circ Z_t\\\nonumber
 		G_{2,t}(\boldsymbol{\theta})&=-U_t(\boldsymbol{\theta})\circ \boldsymbol{F}(X_t\boldsymbol{\gamma}^{\prime})\\\nonumber
 		H_{i,t}(\boldsymbol{\theta})&=-U_t(\boldsymbol{\theta})\circ (\lambda_i\boldsymbol{F}^{\prime}(X_t\boldsymbol{\gamma}_i)\circ X_t)\quad i=1,\ldots,h\\
 		K_{i,t}(\boldsymbol{\theta}) &= [V_t(\boldsymbol{\theta})\circ \boldsymbol{F}^{\prime}(X_t\boldsymbol{\gamma}^{\prime})]^{\prime}X_t \circ e_{i}\quad i=1,\ldots,h\quad k = 1,\ldots, h\\
 		e_{i,k} & = \left\{\begin{array}{ll}
 		1&k=i\\
 		0& k\neq i
 		\end{array}\right.
 	\end{align*}}%
Since $\tilde{\boldsymbol{\theta}}_{n,T}$ is between $\hat{\boldsymbol{\theta}}_{n,T}$ and $\boldsymbol{\theta}_0$, $\hat{\boldsymbol{\theta}}_{n,T}\xrightarrow{p}\boldsymbol{\theta}_0$ so $\tilde{\boldsymbol{\theta}}_{n,T}$ also converges to $\boldsymbol{\theta}_0$ in probability as $n\rightarrow\infty$.
By Assumption \ref{error-dominance}, $\left|\frac{f^{\prime}(s)}{f(s)}\right|,\left|\frac{f^{\prime\prime}(s)}{f(s)}\right|$ and $\left|\frac{f^{\prime 2}(s)}{f^{2}(s)}\right|$ are continuous and are bounded by $a_1+a_2\left|s\right|^{c_1}$ so $U_t(\boldsymbol{\theta}), V_t(\boldsymbol{\theta})$ are continuous. With $\phi_0\in(-\frac{1}{\tau},\frac{1}{\tau})$, $tr((W_nA_0^{-1})^2)= \sum_{i=1}^{n}\frac{\tau_i^2}{(1-\phi_0 \tau_i)^2}$ is also a continuous function of $\phi_0$.

Therefore elements in $\frac{1}{nT}\frac{\partial^2\mathcal{L}_{n,T}(\boldsymbol{\theta})}{\partial \boldsymbol{\theta}\partial\boldsymbol{\theta}^{\prime}}$ are continuous functions for $\boldsymbol{\theta}$ in $\boldsymbol{\Theta}$. Then by the continuity,
{\small\setlength{\abovedisplayskip}{3pt}
	\setlength{\belowdisplayskip}{\abovedisplayskip}
	\setlength{\abovedisplayshortskip}{0pt}
	\setlength{\belowdisplayshortskip}{3pt}
	\begin{gather}\label{cvg-function-p}
	\frac{1}{nT}\frac{\partial^2\mathcal{L}_{n,T}(\tilde{\boldsymbol{\theta}}_{n,T})}{\partial \boldsymbol{\theta}\partial\boldsymbol{\theta}^{\prime}}-\frac{1}{nT}\frac{\partial^2\mathcal{L}_{n,T}(\boldsymbol{\theta}_0)}{\partial \boldsymbol{\theta}\partial\boldsymbol{\theta}^{\prime}}\xrightarrow{p}0,\quad\text{as } \tilde{\boldsymbol{\theta}}_{n,T}\xrightarrow{p}\boldsymbol{\theta}_0
	\end{gather}} %

Finally we will prove that $\left|\frac{1}{nT}\frac{\partial^2\mathcal{L}_{n,T}(\boldsymbol{\theta}_0)}{\partial \boldsymbol{\theta}\partial\boldsymbol{\theta}^{\prime}}-\mathbb{E} \frac{1}{nT}\frac{\partial^2\mathcal{L}_{n,T}(\boldsymbol{\theta}_0)}{\partial \boldsymbol{\theta}\partial\boldsymbol{\theta}^{\prime}}\right|\xrightarrow{p}0$. Since $\ln|A_0|$ can be decomposed as $\sum_{i=1}^{n}\ln (1-\phi_0\tau_{i})$, to show $\mathbb{E} \frac{1}{nT}\frac{\partial^2\mathcal{L}_{n,T}(\boldsymbol{\theta}_0)}{\partial \boldsymbol{\theta}\partial\boldsymbol{\theta}^{\prime}}<\infty$ is equivalent to show
{\small\setlength{\abovedisplayskip}{3pt}
	\setlength{\belowdisplayskip}{\abovedisplayskip}
	\setlength{\abovedisplayshortskip}{0pt}
	\setlength{\belowdisplayshortskip}{3pt}
	\begin{gather}\label{second-derivative-expectation}
	\mathbb{E}\left|\frac{\partial^2}{\partial\boldsymbol{\theta}\partial\boldsymbol{\theta}^{\prime}}\left(\frac{1}{nT}\sum_{t=1}^{T}\sum_{s=1}^{n}\ln(1-\phi_{00}\tau_s)+\ln f(\varepsilon_{s,t}(\boldsymbol{\theta}_0))\right)\right| <\infty
	\end{gather}} %
We first discuss the second derivative with respect to $\phi_0$ component in (\ref{second-derivative-expectation}).
By triangular inequality, $	\mathbb{E}\left|\frac{\partial^2}{\partial\phi_0\partial\phi_0}\frac{1}{nT}\sum_{t=1}^{T}\sum_{s=1}^{n}\Big(\ln(1-\phi_{00}\tau_s)+\ln f(\varepsilon_{s,t}(\boldsymbol{\theta}_0))\Big)\right|<\mathbb{E}\left|\frac{1}{n}\sum_{s=1}^{n}\frac{\partial^2\ln(1-\phi_{00}\tau_s)}{\partial\phi_0\partial\phi_0}\right|+\mathbb{E}\left|\frac{\partial^2 \ln f(\varepsilon_{s,t}(\boldsymbol{\theta}_0))}{\partial\phi_0\partial \phi_0}\right|$ where $\phi_{00}$ is the true value of $\phi_0$. Consider $\mathbb{E}\left|\frac{1}{n}\sum_{s=1}^{n}\frac{\partial^2\ln(1-\phi_{00}\tau_i)}{\partial\phi_0\partial\phi_0}\right|+\mathbb{E}\left|\frac{\partial^2 \ln f(\varepsilon_{s,t}(\boldsymbol{\theta}_0))}{\partial\phi_0\partial\phi_0}\right|$, under stationarity, it can be simplified as
{\small\setlength{\abovedisplayskip}{3pt}
	\setlength{\belowdisplayskip}{\abovedisplayskip}
	\setlength{\abovedisplayshortskip}{0pt}
	\setlength{\belowdisplayshortskip}{3pt}
	\begin{gather}\label{rho-rho-abs}
	\frac{1}{n}tr(W_nA_0^{-1})^2+\mathbb{E}\left|\left(\frac{f^{\prime^2}(\varepsilon_{s,t})}{f^2(\varepsilon_{s,t})}-\frac{f^{\prime\prime}(\varepsilon_{s,t})}{f(\varepsilon_{s,t})}\right)\left(\sum_{k=1}^{n}w_{sk}y_{k,t}\right)^2\right|
	\end{gather}} %
Define $M_n =\{m_{i,j}\}=W_nA_0^{-1}$ and by assumptions, $M_n$ is uniformly bounded in row and column. Suppose the row sum or column sum of $M_n$ is bounded by a constant $b$. We know $\frac{1}{n}tr(W_nA_0^{-1})^2<\infty$. 
So we only need to show $\mathbb{E}\left|\left(\frac{f^{\prime^2}(\varepsilon_{s,t})}{f^2(\varepsilon_{s,t})}-\frac{f^{\prime\prime}(\varepsilon_{s,t})}{f(\varepsilon_{s,t})}\right)\left(\sum_{k=1}^{n}w_{sk}y_{k,t}\right)^2\right|<\infty$.

By simple linear algebra, 
{\small\setlength{\abovedisplayskip}{3pt}
	\setlength{\belowdisplayskip}{\abovedisplayskip}
	\setlength{\abovedisplayshortskip}{0pt}
	\setlength{\belowdisplayshortskip}{3pt}
	\begin{align*}
Y_t &= \sum_{j=0}^{\infty}\Psi_jA_0^{-1}(X_{t-j}\beta+\boldsymbol{F}(X_{t-j}\boldsymbol{\gamma}^{\prime})+\boldsymbol{\varepsilon}_{t-j})\\
&=\sum_{j=0}^{\infty}\Psi_jA_0^{-1}(\boldsymbol{g}(X_{t-j},\boldsymbol{\theta}_0)+\boldsymbol{\varepsilon}_{t-j})
\end{align*}}%
So $W_nY_t = W_n\sum_{j=0}^{\infty}\Psi_jA_0^{-1}(\boldsymbol{g}(X_{t-j},\boldsymbol{\theta}_0)+\boldsymbol{\varepsilon}_{t-j})$. Therefore $(\sum_{k=1}^{n}w_{sk}y_{k,t})^2$ is the $s^{th}$ component of $(W_nY_t\circ W_nY_t)$ and we expand $(W_nY_t\circ W_nY_t)_{s}=$
{\small\setlength{\abovedisplayskip}{3pt}
	\setlength{\belowdisplayskip}{\abovedisplayskip}
	\setlength{\abovedisplayshortskip}{0pt}
	\setlength{\belowdisplayshortskip}{3pt}
	\begin{align}\label{sq-y-1}
	&\left[W_n\sum_{j=0}^{\infty}\Psi_jA_0^{-1}\boldsymbol{g}(X_{t-j},\boldsymbol{\theta}_0) \circ W_n\sum_{j=0}^{\infty}\Psi_jA_0^{-1}\boldsymbol{g}(X_{t-j},\boldsymbol{\theta}_0)\right]_s\\\label{sq-y-2}
	+&\left[2W_n\sum_{j=0}^{\infty}\Psi_jA_0^{-1}\boldsymbol{g}(X_{t-j},\boldsymbol{\theta}_0) \circ W_n\sum_{j=0}^{\infty}\Psi_jA_0^{-1}\boldsymbol{\varepsilon}_{t-j}\right]_s\\\label{sq-y-3}
	+&\left[W_n\sum_{j=0}^{\infty}\Psi_jA_0^{-1}\boldsymbol{\varepsilon}_{t-j}\circ W_n\sum_{j=0}^{\infty}\Psi_jA_0^{-1}\boldsymbol{\varepsilon}_{t-j}\right]_s
	\end{align}}%

From assumptions \ref{ub-weight-matrix} and \ref{ub-backshift-matrix}, we know that $W_n$ is uniformly bounded and $\sum_{j=0}^{\infty}\Psi_jA_0^{-1}$ is absolute summable so $(\ref{sq-y-1})<\infty$ under the stationary condition of $X_t$. Hence, $\mathbb{E}\left|\left(\frac{f^{\prime^2}(\varepsilon_{s,t})}{f^2(\varepsilon_{s,t})}-\frac{f^{\prime\prime}(\varepsilon_{s,t})}{f(\varepsilon_{s,t})}\right)\cdot \text{(\ref{sq-y-1})}\right|<\infty$.

For (\ref{sq-y-2}), when $j>0$, $\boldsymbol{\varepsilon}_{t-j}$ is independent from $\boldsymbol{\varepsilon}_{t}$. So for all $k$ when $j>0$,
{\small\setlength{\abovedisplayskip}{3pt}
	\setlength{\belowdisplayskip}{\abovedisplayskip}
	\setlength{\abovedisplayshortskip}{0pt}
	\setlength{\belowdisplayshortskip}{3pt}
	\begin{gather*}
	\mathbb{E}\left|\left(\frac{f^{\prime^2}(\varepsilon_{s,t})}{f^2(\varepsilon_{s,t})}-\frac{f^{\prime\prime}(\varepsilon_{s,t})}{f(\varepsilon_{s,t})}\right)\cdot \left[W_n\sum_{j=1}^{\infty}\Psi_jA_0^{-1}\boldsymbol{\varepsilon}_{t-j}\right]_s\right|
	=k_1\cdot \left[W_n\sum_{j=1}^{\infty}\Psi_jA_0^{-1}\mathbb{E}\left|\boldsymbol{\varepsilon}_{t-j}\right|\right]_s<\infty
	\end{gather*}}%
when $j=0$, this reduces to $	\mathbb{E}\left|\left(\frac{f^{\prime^2}(\varepsilon_{s,t})}{f^2(\varepsilon_{s,t})}-\frac{f^{\prime\prime}(\varepsilon_{s,t})}{f(\varepsilon_{s,t})}\right)\cdot \left[W_nA_0^{-1}\boldsymbol{\varepsilon}_{t}\right]_s\right|< k_1|b-m_{ss}|+k_2|m_{ss}|$. So $\mathbb{E}\left|\left(\frac{f^{\prime^2}(\varepsilon_{s,t})}{f^2(\varepsilon_{s,t})}-\frac{f^{\prime\prime}(\varepsilon_{s,t})}{f(\varepsilon_{s,t})}\right)\cdot \text{(\ref{sq-y-2})}\right|<\infty$.

For (\ref{sq-y-3}), similar to (\ref{sq-y-2}), we can have
$\mathbb{E}\left|\left(\frac{f^{\prime^2}(\varepsilon_{s,t})}{f^2(\varepsilon_{s,t})}-\frac{f^{\prime\prime}(\varepsilon_{s,t})}{f(\varepsilon_{s,t})}\right)\cdot \text{(\ref{sq-y-3})}\right|<Constant \cdot (k_2+k_3+\mathbb{E}|\varepsilon_{s,t}|)<\infty$.

Therefore combining all these components together,  $\mathbb{E}\left|\left(\frac{f^{\prime^2}(\varepsilon_{s,t})}{f^2(\varepsilon_{s,t})}-\frac{f^{\prime\prime}(\varepsilon_{s,t})}{f(\varepsilon_{s,t})}\right)\left(\sum_{k=1}^{n}w_{sk}y_{k,t}\right)^2\right|=\mathbb{E}\left|\left(\frac{f^{\prime^2}(\varepsilon_{s,t})}{f^2(\varepsilon_{s,t})}-\frac{f^{\prime\prime}(\varepsilon_{s,t})}{f(\varepsilon_{s,t})}\right)\Big((\ref{sq-y-1})+(\ref{sq-y-2})+(\ref{sq-y-3})\Big)\right|<\infty$.
So equation (\ref{rho-rho-abs}) is finite.

Because $\sum_{t=1}^{T}\sum_{s=1}^{n}\ln(1-\phi_{0}\tau_s)$ in (\ref{second-derivative-expectation}) only relates to $\phi_0$, this term goes away when taken second derivative with respect to other parameters. Similar to the proof of $\mathbb{E}\left|\frac{\partial^2 \ln f(\varepsilon_{s,t}(\boldsymbol{\theta}_0))}{\partial\phi_0\partial\phi_0}\right|<\infty$, we can show that $\mathbb{E}\left|\frac{\partial^2 \ln f(\varepsilon_{s,t}(\boldsymbol{\theta}_0))}{\partial\phi_i\partial\phi_j}\right|<\infty$ for $i=0,1,\ldots, p$ and $j=1,\ldots,p$, i.e.,
{\small\setlength{\abovedisplayskip}{3pt}
	\setlength{\belowdisplayskip}{\abovedisplayskip}
	\setlength{\abovedisplayshortskip}{0pt}
	\setlength{\belowdisplayshortskip}{3pt}
	\begin{gather*}
	\mathbb{E}\left|\frac{\partial^2}{\partial\phi_{i}\partial\phi_{j}}\frac{1}{nT}\sum_{t=1}^{T}\sum_{s=1}^{n}\Big(\ln(1-\phi_{00}\tau_s)+\ln f(\varepsilon_{s,t}(\boldsymbol{\theta}_0))\Big)\right| <\infty \quad \text{for } i,j =0,1,\ldots, p
	\end{gather*}} %
Other elements in the matrix (\ref{second-derivative-expectation}) equal to those in $\mathbb{E}\left|\frac{\partial^2 \ln f(\varepsilon_{s,t}(\boldsymbol{\theta}_0))}{\partial\boldsymbol{\theta}\partial \boldsymbol{\theta}^{\prime}}\right|$ and they are also finite.
{\small\setlength{\abovedisplayskip}{3pt}
	\setlength{\belowdisplayskip}{\abovedisplayskip}
	\setlength{\abovedisplayshortskip}{0pt}
	\setlength{\belowdisplayshortskip}{3pt}
	\begin{align}\label{rho-beta-abs}
	\mathbb{E}\left|\frac{\partial^2 \ln f(\varepsilon_{s,t}(\boldsymbol{\theta}_0))}{\partial\phi_{i}\partial \beta^{\prime}}\right| &\leq Constant\cdot|x_{s,t}^{\prime}|
	\left(k_2+k_1\mathbb{E}|\varepsilon_{s,t}|\right)\\
	\mathbb{E}\left|\frac{\partial^2 \ln f(\varepsilon_{s,t}(\boldsymbol{\theta}_0))}{\partial\phi_i\partial \lambda^{\prime}}\right| &\leq  Constant\cdot\boldsymbol{1}_h^{\prime}
	\left(k_2+k_1\mathbb{E}|\varepsilon_{s,t}|\right)\\
	\mathbb{E}\left|\frac{\partial^2 \ln f(\varepsilon_{s,t}(\boldsymbol{\theta}_0))}{\partial\phi_i\partial \gamma_{j}^{\prime}}\right| &\leq Constant\cdot\frac{|\lambda_{j0}x_{s,t}^{\prime}|}{4}\left(k_2+k_1\mathbb{E}|\varepsilon_{s,t}|\right)\\
	\mathbb{E}\left|\frac{\partial^2 \ln f(\varepsilon_{s,t}(\boldsymbol{\theta}_0))}{\partial\beta\partial \beta^{\prime}}\right| &= k_1|x_{s,t}x_{s,t}^{\prime}|\\
	\mathbb{E}\left|\frac{\partial^2 \ln f(\varepsilon_{s,t}(\boldsymbol{\theta}_0))}{\partial\beta\partial \lambda^{\prime}}\right| &= k_1|x_{s,t}\boldsymbol{F}(x_{s,t}^{\prime}\boldsymbol{\gamma}_0)|\\
	\mathbb{E}\left|\frac{\partial^2 \ln f(\varepsilon_{s,t}(\boldsymbol{\theta}_0))}{\partial\beta\partial \boldsymbol{\gamma}_{j}^{\prime}}\right| &\leq \frac{k_1}{4}|\lambda_{j0}x_{s,t}x_{s,t}^{\prime}|\\
	\mathbb{E}\left|\frac{\partial^2 \ln f(\varepsilon_{s,t}(\boldsymbol{\theta}_0))}{\partial\lambda\partial \lambda^{\prime}}\right| & = k_1\left|\boldsymbol{F}(x_{s,t}^{\prime}\boldsymbol{\gamma}_0)^{\prime}\boldsymbol{F}(x_{s,t}^{\prime}\boldsymbol{\gamma}_0)\right| \leq k_1\cdot \boldsymbol{1}_{h\times h}\\
	\mathbb{E}\left|\frac{\partial^2 \ln f(\varepsilon_{s,t}(\boldsymbol{\theta}_0))}{\partial\lambda\partial \boldsymbol{\gamma}_{j}^{\prime}}\right| & = \frac{k_1}{4}|\lambda_{j0}F^{\prime}(x_{s,t}^{\prime}\boldsymbol{\gamma}_{j0})|\cdot|\boldsymbol{F}(x_{s,t}^{\prime}\boldsymbol{\gamma}_{0})^{\prime}x_{s,t}^{\prime}|\leq \frac{k_1|\lambda_{j0}|}{4}\cdot|\boldsymbol{F}(x_{s,t}^{\prime}\boldsymbol{\gamma}_{0})^{\prime}x_{s,t}^{\prime}|\\
	\mathbb{E}\left|\frac{\partial^2 \ln f(\varepsilon_{s,t}(\boldsymbol{\theta}_0))}{\partial\boldsymbol{\gamma}_{k}\partial \boldsymbol{\gamma}_{j}^{\prime}}\right| & \leq \frac{k_1|\lambda_{k0}\lambda_{j0}|}{16}\cdot|x_{s,t}x_{s,t}^{\prime}|,\quad k\neq j\\\label{gamma-gamma-abs}
	\mathbb{E}\left|\frac{\partial^2 \ln f(\varepsilon_{s,t}(\boldsymbol{\theta}_0))}{\partial\boldsymbol{\gamma}_{j}\partial \boldsymbol{\gamma}_{j}^{\prime}}\right| & \leq \frac{k_1\lambda^2_{j0}}{16}\cdot|x_{s,t}x_{s,t}^{\prime}| +\frac{\sqrt{3}k_0|\lambda_{j0}|}{18}|x_{s,t}x_{s,t}^{\prime}|
	\end{align}} %
Then we can apply the ergodic theorem \cite{birkhoff1931proof} and conclude that
{\small\setlength{\abovedisplayskip}{3pt}
	\setlength{\belowdisplayskip}{\abovedisplayskip}
	\setlength{\abovedisplayshortskip}{0pt}
	\setlength{\belowdisplayshortskip}{3pt}	
	\begin{gather}\label{cvg-3}
	\left|\frac{1}{nT}\frac{\partial^2 \mathcal{L}_{n,T}(\boldsymbol{\theta}_0)}{\partial \boldsymbol{\theta}\partial \boldsymbol{\theta}^{\prime}}-\mathbb{E}\frac{1}{nT}\frac{\partial^2 \mathcal{L}_{n,T}(\boldsymbol{\theta}_0)}{\partial \boldsymbol{\theta}\partial \boldsymbol{\theta}^{\prime}}\right|\xrightarrow{p}0
	\end{gather}}%
Recall the equation (\ref{taylor}), we have proved that $\frac{1}{\sqrt{nT}}\frac{\partial \mathcal{L}_{n,T}(\boldsymbol{\theta}_0)}{\partial \boldsymbol{\theta}}$ has the limiting distribution $N(\boldsymbol{0}, B(\boldsymbol{\theta}_0))$. With (\ref{cvg-3}), for $\tilde{\boldsymbol{\theta}}_{n,T}$ between $\hat{\boldsymbol{\theta}}_{n,T}$ and $\boldsymbol{\theta}_0$,  $-\frac{1}{nT}\frac{\partial^2\mathcal{L}_{n,T}(\tilde{\boldsymbol{\theta}}_{n,T})}{\partial \boldsymbol{\theta}\partial \boldsymbol{\theta}^{\prime}}\xrightarrow{p}A(\boldsymbol{\theta}_0)$ so we can conclude that $\sqrt{nT} (\hat{\boldsymbol{\theta}}_{n,T}-\boldsymbol{\theta}_0) \xrightarrow{d} N (\boldsymbol{0}, \boldsymbol{\Omega}_0)$, where $\boldsymbol{\Omega}_0 = A^{-1}(\boldsymbol{\theta}_0)B(\boldsymbol{\theta}_0)A^{-1}(\boldsymbol{\theta}_0)$.
\end{proof}

\section{Numerical Results}
\subsection{Simulation Study}
In this section, we conduct simulation experiments to examine the estimators' behavior for finite samples. We look at two PSTAR-ANN$(1)$ models with one and two neurons with model parameters specified below:
{\small\setlength{\abovedisplayskip}{3pt}
	\setlength{\belowdisplayskip}{\abovedisplayskip}
	\setlength{\abovedisplayshortskip}{0pt}
	\setlength{\belowdisplayshortskip}{3pt}
	\begin{gather}\label{simulation-model}
	Y_{t} = \phi_{0}W_nY_{t} + \phi_{1}W_nY_{t-1} + \boldsymbol{F}(X_t\boldsymbol{\gamma}^{\prime})\lambda +\boldsymbol{\varepsilon}_t\\\nonumber
	\phi_{0}=0.6,\quad \phi_{1}=-0.274,\quad \lambda=1.5\\\nonumber
	\boldsymbol{\gamma} = (\gamma_1,\gamma_2)^{\prime} =(0.75,-0.35)^{\prime}\\\label{simulation-model-2}
	Y_{t} = \phi_{0}W_nY_{t} + \phi_{1}W_nY_{t-1} + X_t\beta+ \boldsymbol{F}(X_t\boldsymbol{\gamma}_1^{\prime})\lambda_1 + \boldsymbol{F}(X_t\boldsymbol{\gamma}_2^{\prime})\lambda_2+\boldsymbol{\varepsilon}_t\\\nonumber
		\phi_{0}=0.6,\quad \phi_{1}=-0.274,\quad \beta = (0.24,-0.7)^{\prime}\\\nonumber
	\lambda_1 = 2,\quad \boldsymbol{\gamma}_1 = (\gamma_{11},\gamma_{12})^{\prime} =(0.75,-0.35)^{\prime}\\\nonumber
		\lambda_2 = 0.8,\quad\boldsymbol{\gamma}_2 = (\gamma_{21},\gamma_{22})^{\prime} =(0.35,-0.5)^{\prime}
	\end{gather}} %
Simulations are conducted in a 30 by 30 lattice grid, so $n= 900$ and $p=1$, $T=30$. Random errors are sampled respectively from three distributions (standard normal, rescaled t-distribution and Laplace distribution) with variance 1.
We generated data for two exogenous variables, observed at different time points $t$ and location $s$. Let $X_t=\begin{pmatrix}
x_{11,t}&\ldots&x_{1n,t}\\
x_{21,t}&\ldots&x_{2n,t}
\end{pmatrix}^{\prime}$.
Usually we would like to normalize predictors before fitting a neural network model to avoid the computation overflow \cite{medeiros2006building} so values of $x_{s,t}[i], i=1,2$, were generated independently from normal distributions $N(0,1.5^2)$ and $N(0,3^2)$ respectively.
The log-likelihood function $\mathcal{L}_{n,T}(\boldsymbol{\theta})$ is given in (\ref{likelihood-simulation}) and we use L-BFGS-B method \cite{byrd1995limited,zhu1997algorithm} (recommended for bound constrained optimization) to find the parameter estimates $\hat{\boldsymbol{\theta}}$ which maximize (\ref{likelihood-simulation}).
{\small\setlength{\abovedisplayskip}{3pt}
	\setlength{\belowdisplayskip}{\abovedisplayskip}
	\setlength{\abovedisplayshortskip}{0pt}
	\setlength{\belowdisplayshortskip}{3pt}
	\begin{align}\label{likelihood-simulation}
	\mathcal{L}_{n,T}(\boldsymbol{\theta}) &= T\ln |I_n-\phi_0 W_n|+\sum_{t=1}^{T}\sum_{s=1}^{n}\ln f(\varepsilon_{s,t}(\boldsymbol{\theta}))\\\nonumber
	\text{for model (\ref{simulation-model}): }&\varepsilon_{s,t}(\boldsymbol{\theta}) = y_{s,t}-\sum_{i=0}^{p}\sum_{k=1}^{n}\phi_iw_{sk}y_{k,t-i}- \boldsymbol{F}(x_{s,t}^{\prime}\boldsymbol{\gamma})\lambda\\\nonumber
	\text{for model (\ref{simulation-model-2}): }&\varepsilon_{s,t}(\boldsymbol{\theta}) = y_{s,t}-\sum_{i=0}^{p}\sum_{k=1}^{n}\phi_iw_{sk}y_{k,t-i}- \boldsymbol{F}(x_{s,t}^{\prime}\boldsymbol{\gamma}_1)\lambda_1-\boldsymbol{F}(x_{s,t}^{\prime}\boldsymbol{\gamma}_2)\lambda_2
	\end{align}} %
For the models under consideration, we estimated the covariance of the asymptotic normal distribution equation (\ref{asymptotic}). 
Since matrices $A(\boldsymbol{\theta}_0)$ and $B(\boldsymbol{\theta}_0)$ involve expected values with respect to the true parameter $\boldsymbol{\theta}_0$, given merely observations, in practice they can be estimated as follows:
{\small\setlength{\abovedisplayskip}{3pt}
	\setlength{\belowdisplayskip}{\abovedisplayskip}
	\setlength{\abovedisplayshortskip}{0pt}
	\setlength{\belowdisplayshortskip}{3pt}
	\begin{align*}
	\hat{A}(\boldsymbol{\theta}_0) &= \frac{1}{nT}\sum_{t=1}^{T}\sum_{s=1}^{n}-\frac{\partial^2 l_{s,t}(\boldsymbol{\theta}_0)}{\partial \boldsymbol{\theta}\partial \boldsymbol{\theta}^{\prime}}\\
	\hat{B}(\boldsymbol{\theta}_0) &= \frac{1}{nT}\sum_{t=1}^{T}\sum_{s=1}^{n}\frac{\partial l_{s,t}(\boldsymbol{\theta}_0)}{\partial \boldsymbol{\theta}}\frac{\partial l_{s,t}(\boldsymbol{\theta}_0)}{\partial \boldsymbol{\theta}^{\prime}}
	\end{align*}} %
where
{\small\setlength{\abovedisplayskip}{3pt}
	\setlength{\belowdisplayskip}{\abovedisplayskip}
	\setlength{\abovedisplayshortskip}{0pt}
	\setlength{\belowdisplayshortskip}{3pt}
	\begin{align*}
	l_{s,t}(\boldsymbol{\theta}) = \frac{1}{n}\ln |I_n-\phi_0 W_n|+ \ln f(\varepsilon_{s,t}(\boldsymbol{\theta}))
	\end{align*}} %
Using (\ref{first-order-derivative}) and (\ref{second-order-derivative}), we can calculate $\hat{A}(\boldsymbol{\theta}_0), \hat{B}(\boldsymbol{\theta}_0)$ to assess the asymptotic properties of parameter estimates.
Note that the derivative of the log-likelihood with respect to $\phi_0$ cannot be calculated directly because it requires taking derivative with respect to a log-determinant of $I_n-\phi_0 W_n$. For small sample sizes, we can compute the determinant directly and get the corresponding derivatives; but for large sample sizes, for example a dataset with $n=900$ observations, $W_n$ is a $900\times 900$ weight matrix which makes it impossible to calculate the derivative directly. Since $W_n$ is a square matrix, we can apply the spectral decomposition such that $W_n$ can be expressed in terms of its $n$ eigenvalue-eigenvector pairs in (\ref{matrix-decomposition}).
So we can apply the following approach to calculate the derivative of $\ln |I_n-\phi_0 W_n|$, which greatly reduces the burden of computations (Viton \cite{viton2010notes}).
{\small\setlength{\abovedisplayskip}{3pt}
	\setlength{\belowdisplayskip}{\abovedisplayskip}
	\setlength{\abovedisplayshortskip}{0pt}
	\setlength{\belowdisplayshortskip}{3pt}\begin{align*}
	\ln |I_n-\phi_0 W_n| = \ln \left(\prod_{s = 1}^{n}(1-\phi_0 \tau_i)\right)
	\end{align*}} %
Further the derivatives of the log-likelihood function with respect to $\phi_0$ is
{\small\setlength{\abovedisplayskip}{3pt}
	\setlength{\belowdisplayskip}{\abovedisplayskip}
	\setlength{\abovedisplayshortskip}{0pt}
	\setlength{\belowdisplayshortskip}{3pt}  
	\begin{align*}
	\frac{\partial l_{s,t}(\boldsymbol{\theta})}{\partial \phi_0} &= \frac{1}{n}\sum_{i = 1}^{n}\frac{-\tau_i}{(1-\rho \tau_i)}+\{y_{s,t}-\sum_{i=0}^{p}\phi_i\sum_{j=1}^{n}w_{sj}y_{j,t-i}-\lambda F(x_{s,t}^{\prime}\boldsymbol{\gamma})\}\cdot\left(\sum_{j=1}^{n}w_{sj}y_{j,t}\right)\\
	\frac{\partial^2 l_{s,t}(\boldsymbol{\theta})}{\partial \phi_0\partial \phi_0} & = -\frac{1}{n}\sum_{i = 1}^{n}\left[\frac{\tau_i^2}{(1-\phi_0 \tau_i)^2}+\left(\sum_{j=1}^{n}w_{sj}y_{j,t}\right)^2\right]
	\end{align*}} %
Finally we can estimate the covariance matrix by equation (\ref{asy_var}).
{\small\setlength{\abovedisplayskip}{3pt}
	\setlength{\belowdisplayskip}{\abovedisplayskip}
	\setlength{\abovedisplayshortskip}{0pt}
	\setlength{\belowdisplayshortskip}{3pt}
	\begin{align}\label{asy_var}
	\hat{\boldsymbol{\Omega}} &= \hat{A}^{-1}(\boldsymbol{\theta}_0)\hat{B}(\boldsymbol{\theta}_0)\hat{A}^{-1}(\boldsymbol{\theta}_0)
	\end{align}} %
In each simulation study, we compute $\hat{\boldsymbol{\theta}}$ for each of 200 replicates.  
The estimated $\hat{\boldsymbol{\Omega}}$ of the asymptotic covariance matrix $\hat{\boldsymbol{\Omega}}$ is computed based on a sample with $n=10000, T=100$ simulated observations. Table \ref{simulation-2} compares the empirical mean and standard errors (in parentheses) of $\hat{\boldsymbol{\theta}}$ with the true value and their estimated asymptotic standard deviations. From simulation results of the two models, the empirical standard deviations of $\hat{\boldsymbol{\theta}}$ are close to the asymptotic standard deviations, which implies that the estimators' large finite sample behavior roughly matches their asymptotic distributions. Note that when $\varepsilon_t$ is sampled from a Laplace distribution, this covariance matrix cannot be computed because its second order derivative is not differentiable at $0$. But the simulated $\hat{\boldsymbol{\theta}}$'s still exhibit normal properties. Normal plots for parameter estimates are shown in Figure \ref{qqplot} and give a strong indication of normality.  
\setlength{\extrarowheight}{5pt}
\begin{table}[h!]
	\centering
	\scalebox{0.85}{
	\begin{tabular}{lcccccccccc}
		\toprule
		\multicolumn{11}{c}{Model 1: $Y_{t} = \phi_{0}W_nY_{t} + \phi_{1}W_nY_{t-1} + \boldsymbol{F}(X_t\boldsymbol{\gamma}^{\prime})\lambda +\varepsilon_t$}
		\\\
		&&\multicolumn{2}{c}{$\varepsilon_t$}& $\hat{\phi}_0 $ &$\hat{\phi}_1 $ & $\hat{\lambda} $ & $\hat{\gamma}_1 $ & $\hat{\gamma}_2 $&& \\
		&&\multicolumn{2}{c}{true value}& $0.6$ &$-0.274$ & $1.50$ & $0.75$ & $-0.35$&& \\\midrule
		&&\multicolumn{2}{c}{\multirow{3}{*}{$N(0,1)$}}
		& 0.5997 & -0.2743 &  1.5025 & 0.7485& -0.3476&&\\ [-5pt]
		&&&& (0.0065)& (0.0079)& (0.0274)&(0.0269)& (0.0134)&&\\ [-5pt]
		&&&& [0.0079] & [0.0085] & [0.0308] & [0.0310]&[0.0147]\\
		&&\multicolumn{2}{c}{\multirow{3}{*}{$t(4)$}}& 0.5994& -0.2737 & 1.5000 & 0.7531 & -0.3507\\[-5pt]
		&&&& (0.0059) & (0.0069) & (0.0236) & (0.0249)&(0.0112) \\ [-5pt]
		&&&& [0.0068] & [0.0071] & [0.0259] & [0.0258]&[0.0122]\\
		&&\multicolumn{2}{c}{$Laplace$} & 0.5999 & -0.2736& 1.4992 & 0.7501& -0.3504\\[-5pt]
		&&\multicolumn{2}{c}{$(0, \frac{\sqrt{2}}{2})$}& (0.0048) & (0.0058) & (0.0199) & (0.0196)&(0.0097) \\[5pt]
       \toprule
		\multicolumn{11}{c}{Model 2: $Y_{t} = \phi_{0}W_nY_{t} + \phi_{1}W_nY_{t-1} + X_t\beta+ \boldsymbol{F}(X_t\boldsymbol{\gamma}_1^{\prime})\lambda_1 + \boldsymbol{F}(X_t\boldsymbol{\gamma}_2^{\prime})\lambda_2+\varepsilon_t$}
		\\
		$\varepsilon_t$& $\hat{\phi}_0 $ &$\hat{\phi}_1 $ &\multicolumn{2}{c}{$\hat{\beta}$} &$\hat{\lambda} _1$&\multicolumn{2}{c}{$\hat{\boldsymbol{\gamma}}_1$}&$\hat{\lambda} _2$&\multicolumn{2}{c}{$\hat{\boldsymbol{\gamma}}_2$}\\
		& $0.6 $ &$-0.274$ &$0.24$&$-0.70$&$2$&$0.75$&$0.7$&$0.8$&$0.35$& $-1$\\\midrule
		\multirow{3}{*}{$N(0,1)$}
		& 0.6000 & -0.2748 &  0.2402 & -0.6985& 1.9927& 0.7503 & 0.7030&  0.8076 & 0.3577 & -1.0159\\ [-5pt]
		& (0.0039)& (0.0046)& (0.0137)&(0.0140)& (0.0928)& (0.0962) & (0.0369) &  (0.0450) & (0.0899)& (0.1209)\\ [-5pt]
		& [0.0040] & [0.0044] & [0.0135] & [0.0141]&[0.0921]& [0.0920]& [0.0390] &  [0.0449] & [0.0835]& [0.1243]\\
		\multirow{3}{*}{$t(4)$}& 0.5999& -0.2740 & 0.2402 & -0.7005 & 2.0008 & 0.7496& 0.7016 &  0.7989 & 0.3521& -1.0078\\[-5pt]
		& (0.0036) & (0.0034) & (0.0130) & (0.0106)&(0.0727) & (0.0714) & (0.0332) &  (0.0392) & (0.0749) & (0.0972)\\ [-5pt]
		& [0.0035] & [0.0036] & [0.0116] & [0.0113]&[0.0759]& [0.0324] &  [0.0371]& [0.0758]& [0.0697] & [0.1026]\\
		$Laplace$ & 0.6006 & -0.2743 & 0.2408 & -0.6997& 1.9983 & 0.7509 & 0.7026 &  0.8034 & 0.3477 & -1.0089\\[-5pt]
		$(0, \frac{\sqrt{2}}{2})$ & (0.0030) & (0.0030) & (0.0100) & (0.00995)&(0.0638) & (0.0624) & (0.0256) &  (0.0293) & (0.0621) & (0.0873)\\\bottomrule
	\end{tabular}}
	\caption{Empirical means and standard errors (in parentheses) of parameter estimates when $\varepsilon$ is sampled from a standard normal, standardized student t distribution and a Laplace distribution. The asymptotic standard errors are displayed for reference in square brackets.}
	\label{simulation-2}
\end{table}
\begin{figure}[h!]
	\begin{center}
		\begin{center}
			\includegraphics[width=0.9\linewidth]{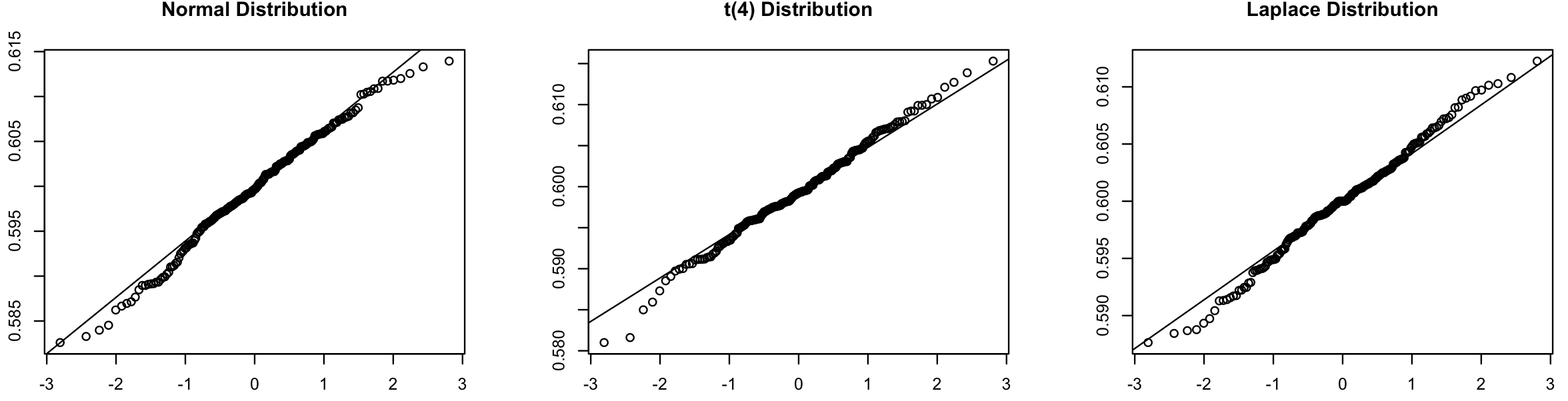}
			\includegraphics[width=0.9\linewidth]{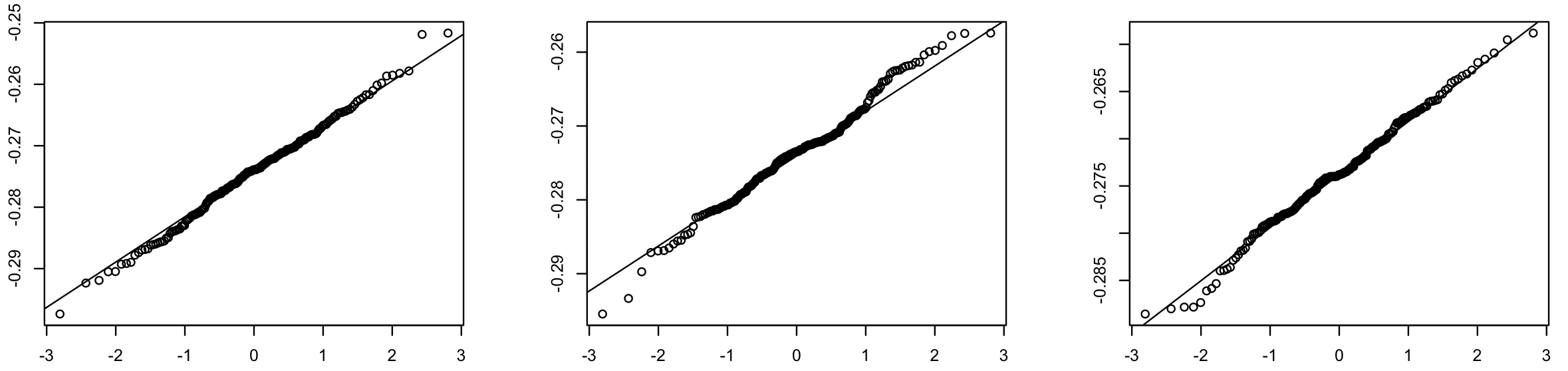}
			\includegraphics[width=0.9\linewidth]{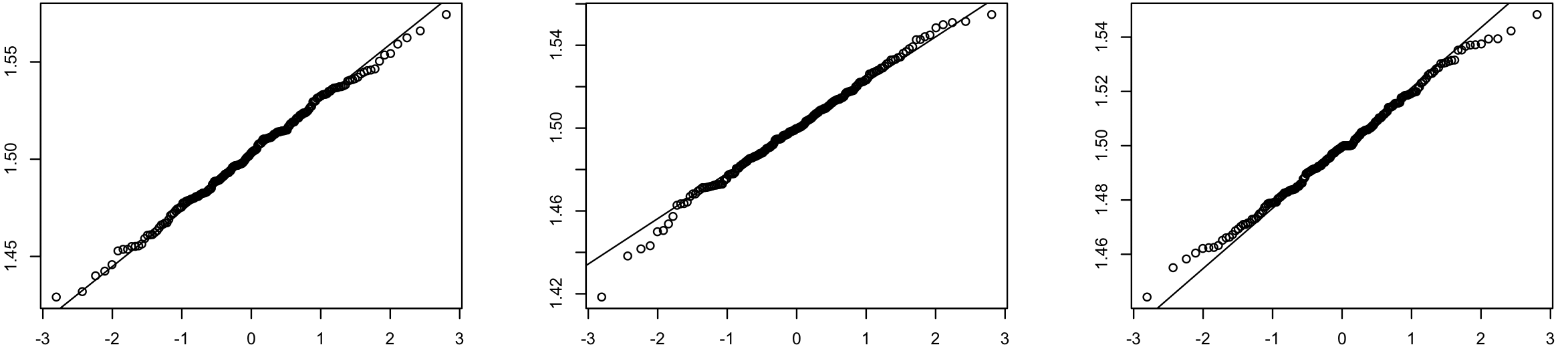}
			\includegraphics[width=0.9\linewidth]{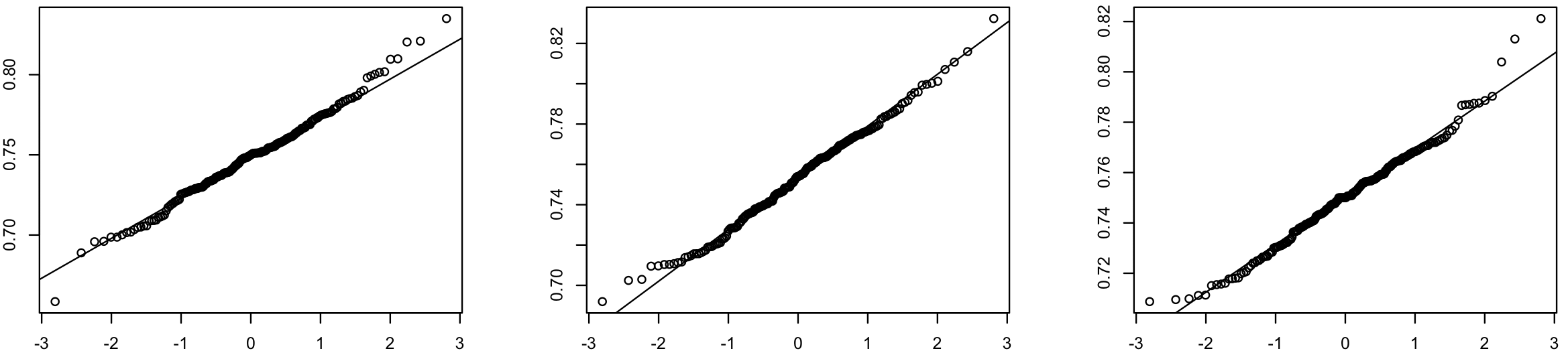}
			\includegraphics[width=0.9\linewidth]{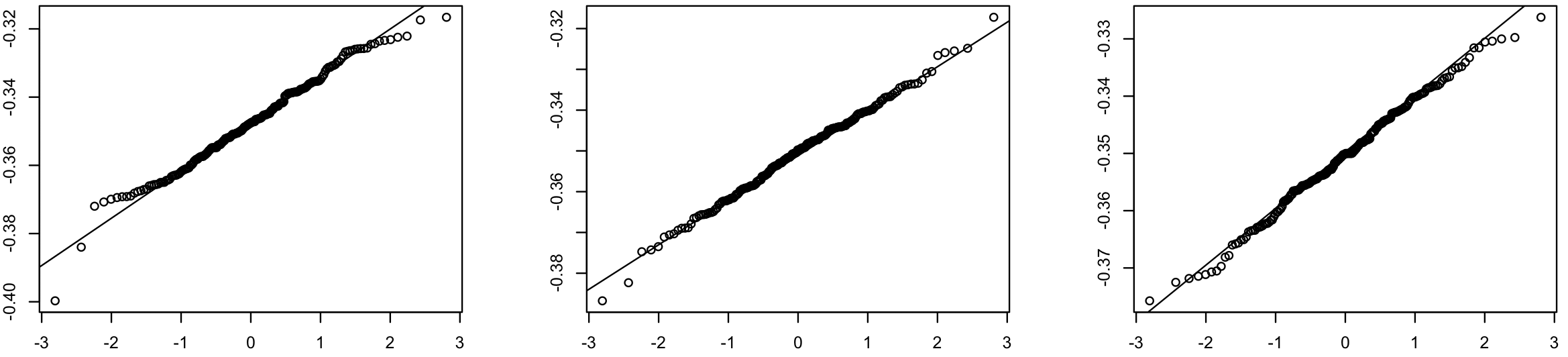}
			\caption{ Normal plots for parameter estimates $\phi_0$ (1st row), $\phi_1$ (2nd row), $\lambda$ (3rd row) and $\gamma_{1}$ (4th row), $\gamma_{2}$ (5th row) of model (\ref{simulation-model}) when $\varepsilon_s$ follows a standard normal distribution (first column), standardized t distribution (middle column) and Laplace distribution (last column) $n = 30\times 30, T=30$ }
			\label{qqplot}
		\end{center}
	\end{center}
\end{figure}
\subsection{Real Data Example}
Spatial models have a lot of applications in understanding spatial interactions in  cross-sectional data. In our first chapter we applied a partially specified spatial autoregressive model to understand the relationships between vote choices and social factors. In this chapter, we want to use a partially specified space time autoregressive model to further analyze the time influence in the electoral dynamics.

We focus on the proportion of votes cast for U.S. presidential candidates at the county level in 2004. Counties are grouped by state, and let $Y_{t}, Y_{t-1}$ (so $t=1,2$, i.e., observe $Y_1$ and $Y_2$) be the corresponding fraction of votes (vote-share) in a county for the Democratic candidate in 2004 and 2000.
Predictors $X_t$ are chosen from economic and social factors covering the living standard, economy development and racial distribution.
\begin{figure}[h!]
	\centering
	\includegraphics[width=0.45\linewidth]{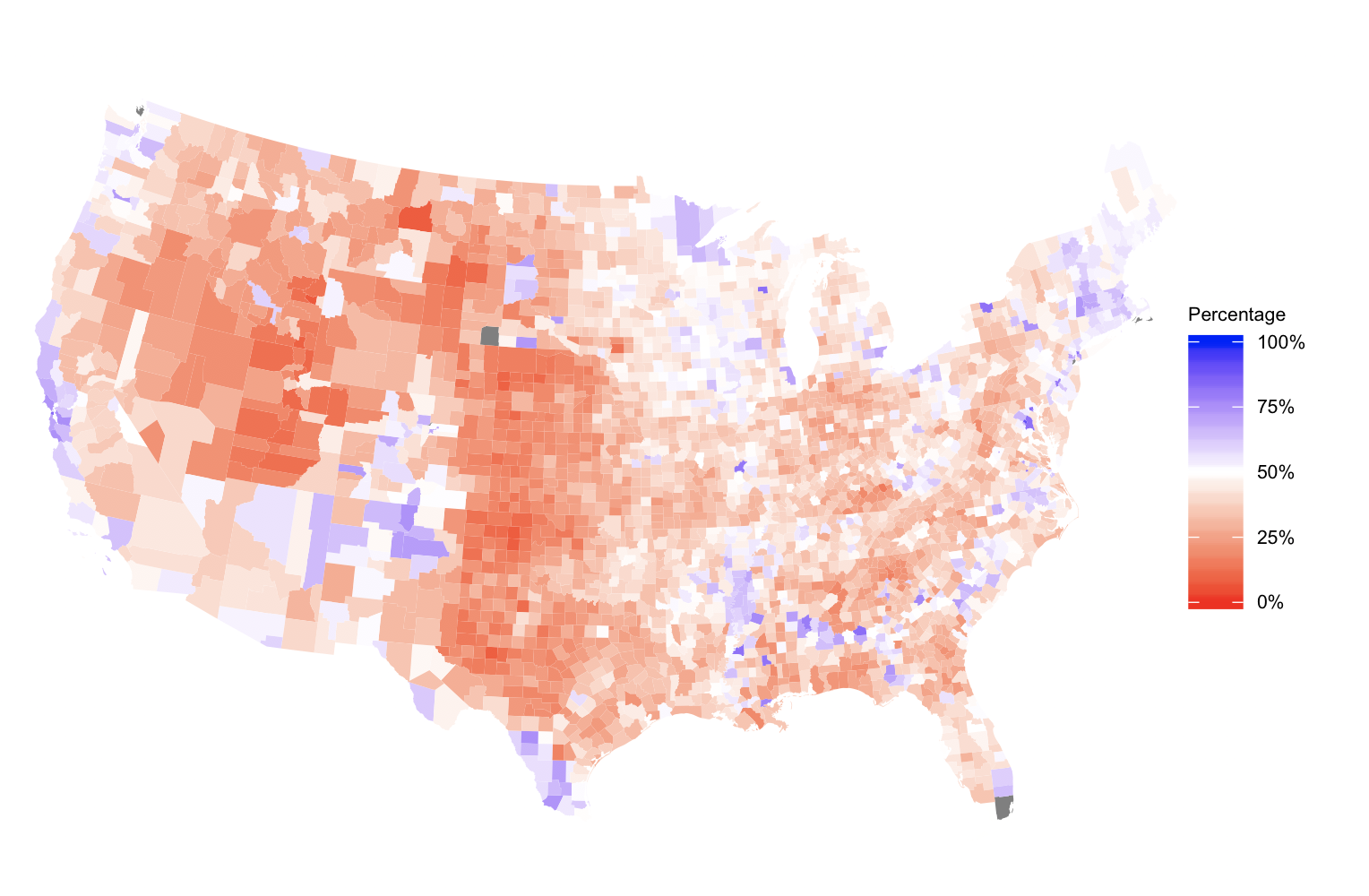}
	\includegraphics[width=0.4\linewidth]{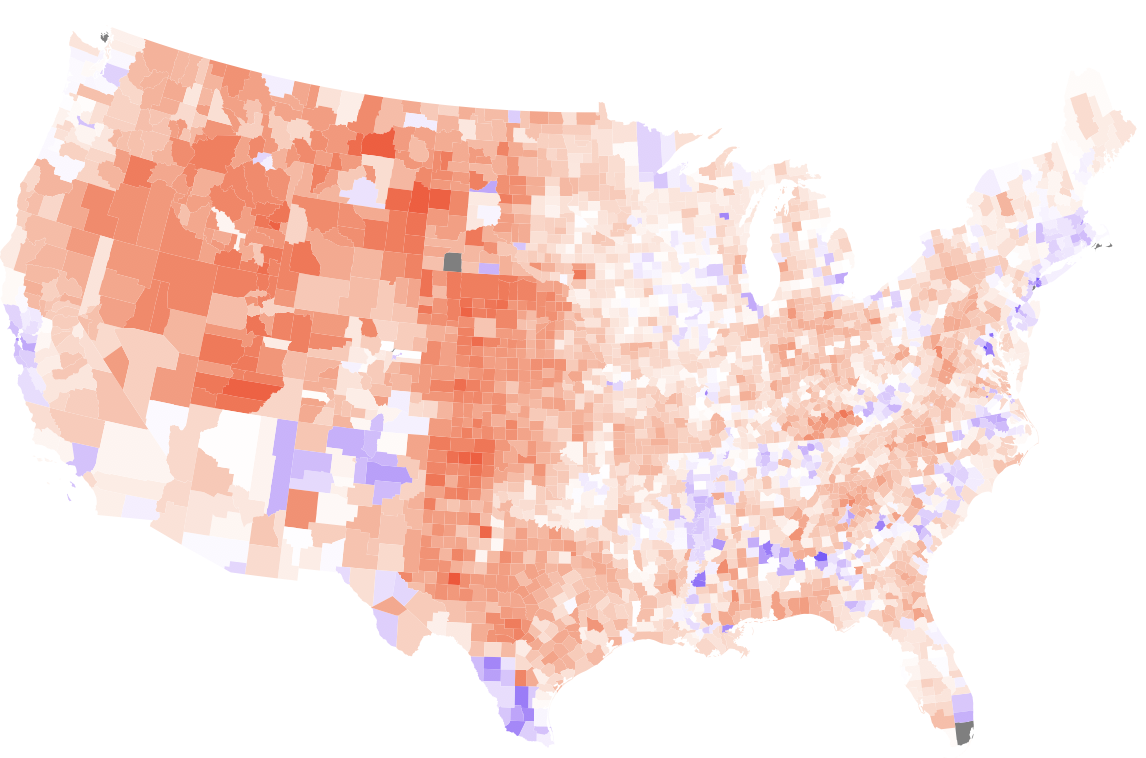}
	\caption{Fractions of vote-shares per county for Democratic presidential candidate in 2004 (left) and 2000 (right)}
	\label{democracy}
\end{figure}
Figure \ref{democracy} shows the observed values of $Y_2$ for 2004 and $Y_1$ for 2000 in a US map. Despite the strong spatial correlation (by Moran's Test on $Y_t$ test statistic = $52.4$, P-value $ <2.2\times10^{-16}$), these heat maps also exhibit the correlation across time since the two heat maps look rather similar. This indicates that $Y_t$, the fraction of vote-share for Democratic candidate, is not independently distributed across the space or time. Therefore we consider fitting a space time model to the data.

In our analysis, we exclude the four U.S. counties with no neighbors (San Juan, Dukes, Nantucket, Richmond) to avoid the non-singularity of our spatial weight matrix $W_n$ in the modeling, so the total number of observations is $n = 3107$. Continuing our analysis in the first chapter, 
the selected explanatory variables are percent residents under 18 years in 2004 $X_{1,t}$ (\texttt{UNDER18}), percent white residents in 2004 $X_{2,t}$ (\texttt{WHITE}), percent residents below poverty line in 2004 $X_{3,t}$ (\texttt{pctpoor}).

We also assume the random error follows a scaled $t(8)$ distribution and, similar to previous chapter, perform variable transformations as follows:
{\small\setlength{\abovedisplayskip}{3pt}
	\setlength{\belowdisplayskip}{\abovedisplayskip}
	\setlength{\abovedisplayshortskip}{0pt}
	\setlength{\belowdisplayshortskip}{3pt}
	\begin{align*}
	Y_t^{\ast} &= Y_t/8\\
	Y_{t-1}^{\ast} & = Y_{t-1}/8\\
	\tilde{X}_{1,t}&= (I_{3107}-0.6W_{3107})X_{1,t}\\
	X_{1,t}^{\ast} &= \frac{\tilde{X}_{1,t}-Average(\tilde{X}_{1,t})}{Std(\tilde{X}_{1,t})}\\
	X_{2,t}^{\ast} &= \frac{X_{2,t}-Average(X_{2,t})}{Std(X_{2,t})}\\
	X_{3,t}^{\ast} &= \frac{X_{3,t}-Average(X_{3,t})}{Std(X_{3,t})}
	\end{align*}} 
Figure \ref{hist-x-y} illustrates histograms of $Y_t^{\ast}$ (first row) and histograms of exogenous variables $X_{1,t}^{\ast}, X_{2,t}^{\ast}, X_{3,t}^{\ast}$ when $t=1,2,3$ ($t=3$ represents the year 2008) respectively. Comparing their histograms at different years, we can observe that the distributions look similar so we may consider $X_t$ and $Y_t$ as stationary processes across time.
\begin{figure}[h!]
	\centering
	\includegraphics[width=0.9\linewidth]{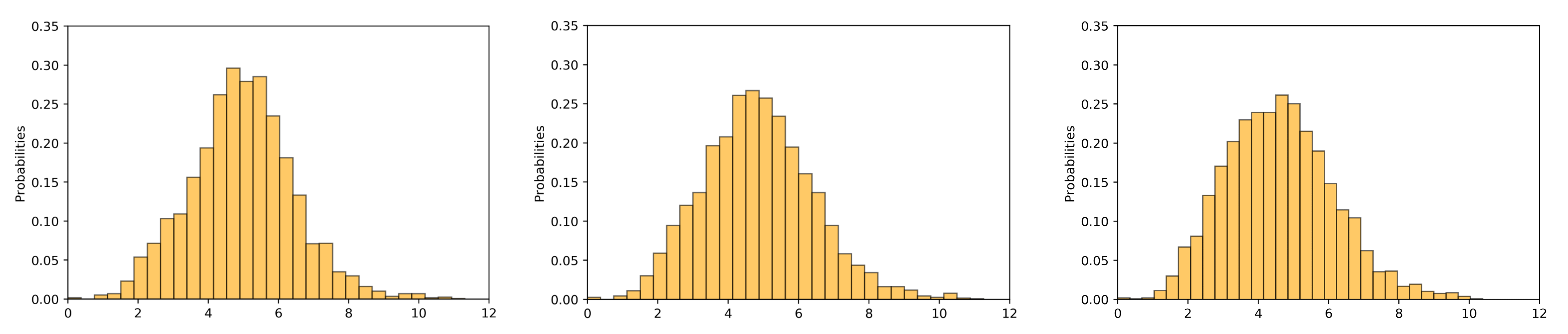}
	\includegraphics[width=0.9\linewidth]{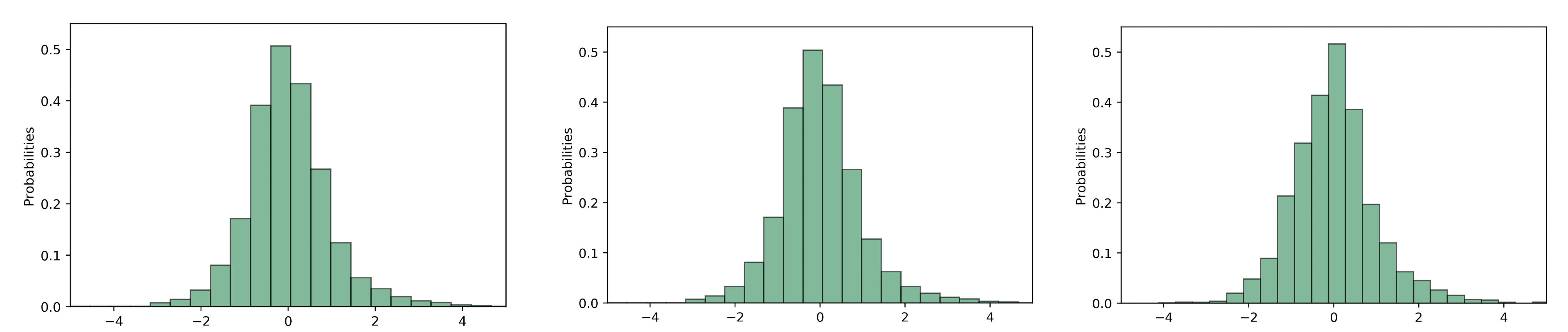}
	\includegraphics[width=0.9\linewidth]{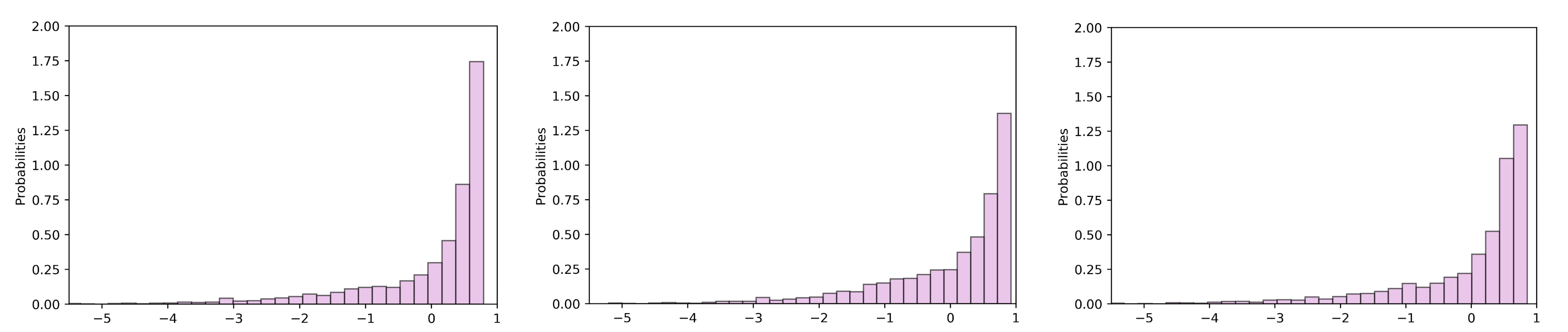}
	\includegraphics[width=0.9\linewidth]{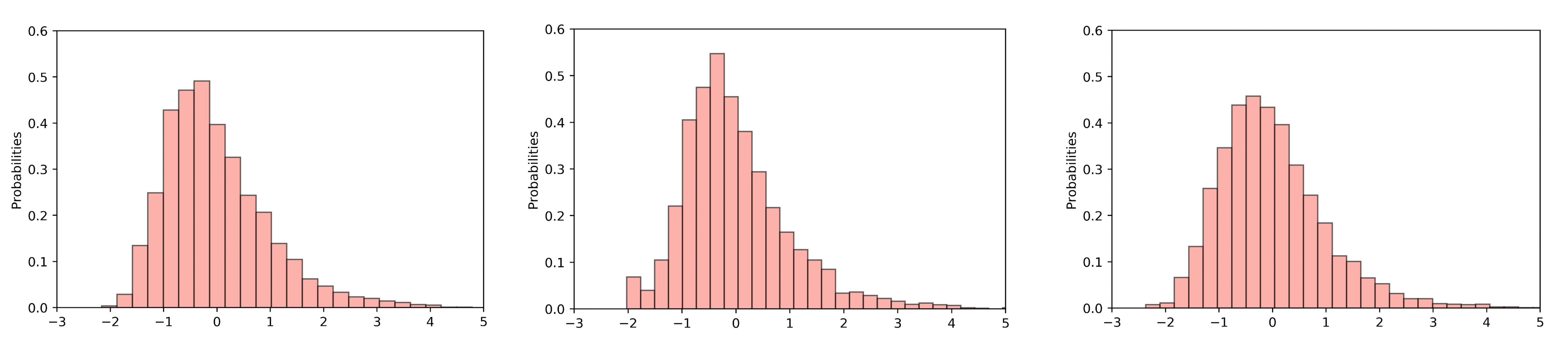}
	\caption{Histograms of $Y_t^{\ast}$ (1st row), $X_{1,t}^{\ast}$ (2nd row), $X_{3,t}^{\ast}$ (3rd row) and $X_{3,t}^{\ast}$ (4th row) for $t=1,2,3$ corresponding to the year 2000 (left), 2004 (middle) and 2008 (right)}
	\label{hist-x-y}
\end{figure}

The estimated PSAR-ANN model in chapter 1 is:
{\small\setlength{\abovedisplayskip}{3pt}
	\setlength{\belowdisplayskip}{\abovedisplayskip}
	\setlength{\abovedisplayshortskip}{0pt}
	\setlength{\belowdisplayshortskip}{3pt}
	\begin{align}\label{psar-ann-model}\nonumber
	Y_t^{\ast} =&0.721W_{3107}Y_t^{\ast} + 1.693
	-0.185X_{1,t}^{\ast}- 0.658X_{2,t}^{\ast}+0.181X_{3,t}^{\ast}\\
	&-0.937F(1.509X_{1,t}^{\ast} -2.544X_{2,t}^{\ast}+2.268X_{3,t}^{\ast})+\hat{\varepsilon}_t
	\end{align}} 
In this chapter we would like to add time into the model and we fit two PSTAR-ANN$(1)$ models with one and two neurons respectively. Similarly we find the parameter estimates by maximizing the corresponding log-likelihood functions and use the L-BFGS-B algorithm to search for the optimum. Detailed optimization steps are similar to those in chapter 1. The model fits are shown below. One is the PSTAR-ANN$(1)$ with one neuron:
{\small\setlength{\abovedisplayskip}{3pt}
	\setlength{\belowdisplayskip}{\abovedisplayskip}
	\setlength{\abovedisplayshortskip}{0pt}
	\setlength{\belowdisplayshortskip}{3pt}
	\begin{align}\label{pstar-oneneuron-model}\nonumber
	Y_t^{\ast} &=0.425W_{3107}Y_t^{\ast} + 0.464W_{3107}Y_{t-1}^{\ast} -1.173
	+0.148X_{1,t}^{\ast} -1.177X_{2,t}^{\ast}-0.153X_{3,t}^{\ast}\\
	&+3.056F(-0.722X_{1,t}^{\ast} +1.689X_{2,t}^{\ast}+0.248X_{3,t}^{\ast})+\hat{\varepsilon}_t
	\end{align}} 
Another is the PSTAR-ANN$(1)$ with two neurons:
{\small\setlength{\abovedisplayskip}{3pt}
	\setlength{\belowdisplayskip}{\abovedisplayskip}
	\setlength{\abovedisplayshortskip}{0pt}
	\setlength{\belowdisplayshortskip}{3pt}
	\begin{align}\label{pstar-twoneuron-model}\nonumber
	Y_t^{\ast} &=0.417W_{3107}Y_t^{\ast} + 0.467W_{3107}Y_{t-1}^{\ast} -1.576
	+0.203X_{1,t}^{\ast} -1.222X_{2,t}^{\ast}-0.057X_{3,t}^{\ast}\\
	&+0.699F(-1.249X_{1,t}^{\ast} +0.084X_{2,t}^{\ast}-3.247X_{3,t}^{\ast})\\\nonumber
	&+3.180F(-0.621X_{1,t}^{\ast} +1.621X_{2,t}^{\ast}+0.495X_{3,t}^{\ast})+\hat{\varepsilon}_t
	\end{align}} 
Comparing  the three models (\ref{psar-ann-model}), (\ref{pstar-oneneuron-model}) and (\ref{pstar-twoneuron-model}), the coefficients estimates are all positive so it is apparent that there exist a positive space correlation, between $y_{s,t}$ and its neighbors, and also a positive time correlation between $Y_t$ and $Y_{t-1}$. The P-values of Moran's test statistic of PSTAR-ANN$(1)$ model residuals (residuals of model (\ref{pstar-oneneuron-model}) and (\ref{pstar-twoneuron-model})) are higher than that of model (\ref{psar-ann-model}), which indicates that PSTAR-ANN$(1)$ models are able to describe more spatial correlations than the PSAR-ANN model. 
\begin{figure}[h!]
	\centering
	\includegraphics[width=0.6\linewidth]{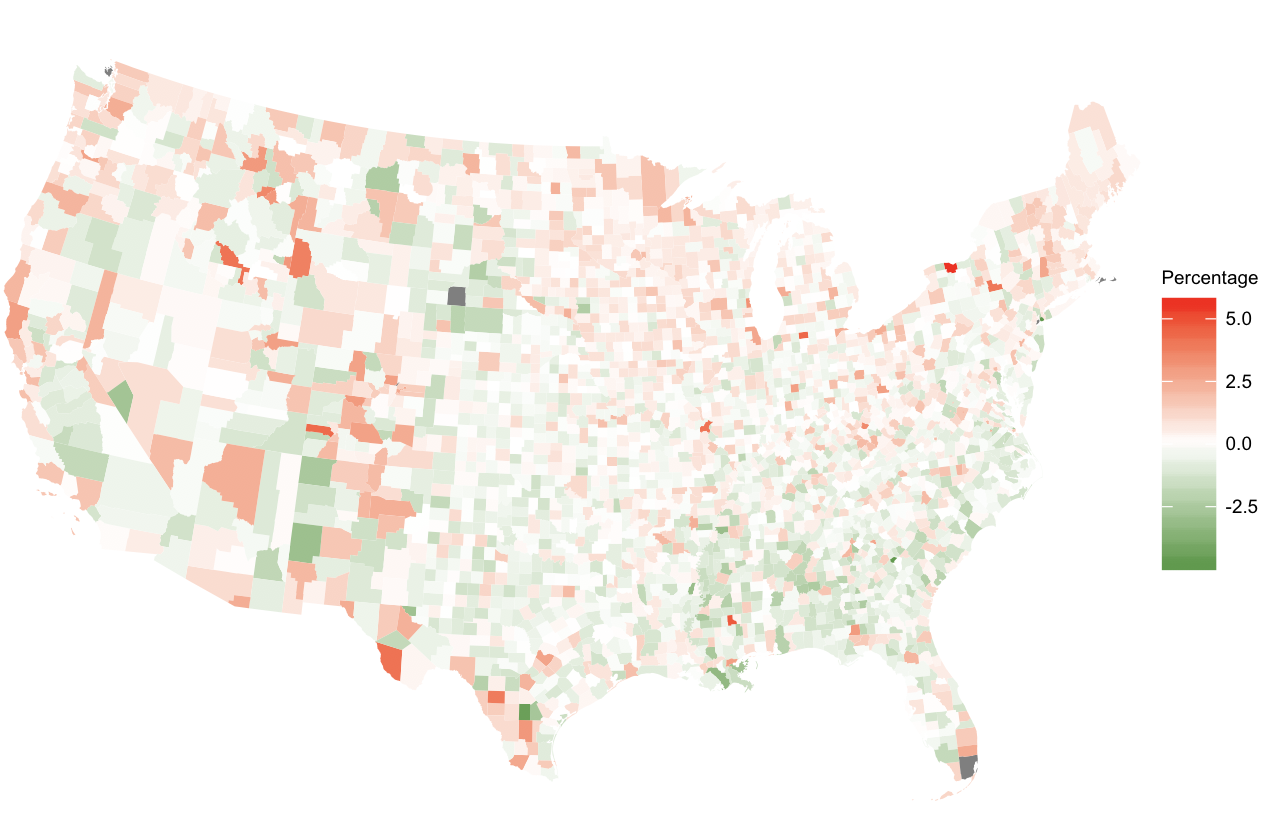}
	\caption{Residuals heat map (calculated from the PSTAR-ANN model with one neuron)}
	\label{res_plots}
\end{figure}
For the preliminary comparison purpose, we compare the AICs (AIC $=2\#\text{parameters}-2\ln L_{n,T}(\hat{\boldsymbol{\theta}})$) of the three models (See table \ref{models-compare}).
For likelihood ratio test ($\mathcal{H}_0$: Model (\ref{psar-ann-model}) is adequate, $\mathcal{H}_1$: model (\ref{pstar-oneneuron-model}) is adequate), the test statistic $-2\ln L_{\text{Model (\ref{psar-ann-model})}}+2\ln L_{\text{Model (\ref{pstar-oneneuron-model})}}= 287.17$ with $df = 6$, P-value $< 0.05$, so we rejected $\mathcal{H}_0$ and conclude that the PSTAR-ANN$(1)$ model with one neuron is a better fit. Similarly we apply the same method to compare the two PSTAR-ANN$(1)$ models and conclude that the model with two neurons is better (the test statistic $-2\ln L_{\text{Model (\ref{pstar-oneneuron-model})}}+2\ln L_{\text{Model (\ref{pstar-twoneuron-model})}} = 40.33$ with $df = 4$, P-value $< 0.05$).
\begin{table}[h]
	\centering
	\begin{tabular}{cccc}
		\toprule
		\multirow{2}{*}{Models}
		& PSAR-ANN & PSTAR-ANN & PSTAR-ANN\\
		&(one neuron)&(one neuron)&(two neurons)\\\midrule
		$\#$ Parameters& $9$ & $10$& $14$\\
		\multirow{2}{*}{Moran's Test}
		& $0.0745$&$0.2336$&$0.3368$\\ [-2pt]
		& $(1.7836)$& $(-1.1910)$& $(-0.9604)$\\
		$-\ln L$& $1879.35$&$1734.5$&$1710.33$\\
		$AIC$&$3776.17$ &$3489$&$3448.669$\\
		\bottomrule     
	\end{tabular}
	\caption{Model Comparisons: PSAR-ANN model with one neuron (\ref{psar-ann-model}), PSTAR-ANN models with one (\ref{pstar-oneneuron-model}) and two neurons (\ref{pstar-twoneuron-model})}
	\label{models-compare}
\end{table}
The covariance matrices for the parameter estimates of model (\ref{pstar-oneneuron-model}) and (\ref{pstar-twoneuron-model}) are calculated and the $95\%$ confidence intervals for the model parameters are shown in Tables \ref{CI} and \ref{CI_PSTAR_ANN_2}. 
\setlength{\extrarowheight}{5pt}
\begin{table}[h!]
	\centering
	\begin{tabular}{cccc}
		\toprule
		Parameter& Estimate&Std.&$95\%$ C.I. \\ \midrule
		$\phi_0$ & 0.425 & 0.0086&  $(0.4081, 0.4419)$ \\ [-5pt]
		$\phi_1$& 0.464& 0.0182& $(0.4283, 0.4997)$\\ [-5pt]
		$\beta_0$& -1.173 & 0.3283 & $(-1.8165, -0.5295)$\\[-5pt]
		$\beta_1$& 0.148 & 0.0697 & $(0.0114, 0.2846)$\\ [-5pt]
		$\beta_2$& -1.177 & 0.1638 & $(-1.4980, -0.8560)$\\ [-5pt]
		$\beta_3$& -0.153 & 0.1079 & $(-0.3645, 0.0585)^{\ast}$\\[-5pt]
		$\lambda$& 3.056& 0.6397& $(1.8022, 4.3098)$\\ [-5pt]
		$\gamma_1$& -0.722& 0.1278& $(-0.9725, -0.4715)$\\ [-5pt]
		$\gamma_2$&1.689& 0.1762& $(1.3436, 2.0344)$\\ [-5pt]
		$\gamma_3$& 0.248& 0.1890& $(-0.1224, 0.6184)^{\ast}$\\
		\bottomrule       
	\end{tabular}
	\caption{Parameter estimates of PSTAR-ANN model (\ref{pstar-oneneuron-model}) parameters with $95\%$ confidence intervals ($\ast$ indicates the insignificance)}
	\label{CI}
\end{table}
From Table \ref{CI}, all the parameters, except $X_{3,t}$ (\texttt{pctpoor}), are significant at 0.05 significance level. Table \ref{CI_PSTAR_ANN_2} shows the $95\%$ level of parameter estimates in model (\ref{pstar-twoneuron-model}).
\setlength{\extrarowheight}{5pt}
\begin{table}[h!]
	\centering
	\begin{tabular}{cccc}
		\toprule
		Parameter& Estimate&Std.&$95\%$ C.I. \\ \midrule
		$\phi_0$ & 0.417 & 0.0086&  $(0.4000, 0.4339)$ \\ [-5pt]
		$\phi_1$& 0.467& 0.0178& $(0.4321, 0.5019)$\\ [-5pt]
		$\beta_0$& -1.576 & 0.3063 & $(-2.1764, -0.9756)$\\[-5pt]
		$\beta_1$& 0.203 & 0.0731 & $(0.0598, 0.3462)$\\ [-5pt]
		$\beta_2$& -1.222 & 0.1507 & $(-1.5174, -0.9266)$\\ [-5pt]
		$\beta_3$& -0.057 & 0.0926 & $(-0.2385, 0.1245)^{\ast}$\\[-5pt]
		$\lambda_1$& 3.180& 1.2624& $(0.7057, 5.6543)$\\ [-5pt]
		$\gamma_{11}$& -0.621& 0.1193& $(-0.8548, -0.3872)$\\ [-5pt]
		$\gamma_{12}$&1.621& 0.1521& $(1.3230, 1.9190)$\\ [-5pt]
		$\gamma_{13}$& 0.495& 0.1063& $(0.2866, 0.7034)$\\[-5pt]
		$\lambda_2$& 0.699& 0.2397& $(0.2291, 1.1689)$\\ [-5pt]
		$\gamma_{21}$& -1.294& 0.6060& $(-2.4368, -0.0612)$\\ [-5pt]
		$\gamma_{22}$&0.084& 0.4859& $(-0.8683, 1.0363)^{\ast}$\\ [-5pt]
		$\gamma_{23}$& -3.247& 0.3570& $(-3.9469, -2.5471)$\\
		\bottomrule       
	\end{tabular}
	\caption{Parameter estimates of PSTAR-ANN model (\ref{pstar-twoneuron-model}) parameters with $95\%$ confidence intervals ($\ast$ indicates the insignificance)}
	\label{CI_PSTAR_ANN_2}
\end{table}

From Table \ref{CI} and \ref{CI_PSTAR_ANN_2}, we can see that values of $y_{s,t}$ are positively spatially correlated in both space and time. 
Looking at the signs of parameter estimates of coefficients, we can see that the sign of variable \texttt{UNDER18} in model (\ref{psar-ann-model}) is negative while positive in model (\ref{pstar-oneneuron-model}) and (\ref{pstar-twoneuron-model}). Considering its parameter estimate significant in all models, this indicates that age and vote-shares for Democratic candidates can be dependent but the percent residents under 18 may not be a good measurement for this social factor. We should consider using other age related variables to predict $Y_t$ such as the percent young voters between 18 and 30 years old.
Variable \texttt{WHITE} is negatively correlated with $Y_t$ in all three fitted models and this negative correlation accords with our common sense that white voters tend to support the Republican candidate.
The last variable \texttt{pctpoor} is bit tricky because it is not significant in model (\ref{pstar-oneneuron-model}) but is significant in the neural network component in model (\ref{pstar-twoneuron-model}). Regarding to this, it needs further assessment to decide if \texttt{pctpoor} should be included in the model.
In chapter 3, we will further discuss the model selection in detail. To conclude, our proposed model PSTAR-ANN appears to successfully capture some presidential election dynamics over both space and time. It allows for non-Gaussian random errors and is flexible in learning nonlinear relationships between the response and exogenous variables.

\small{
	\singlespacing
	\bibliography{references}
}
\end{document}